\documentclass[aps,pre,reprint,amsmath,amsfonts,amssymb]{revtex4-1}
\usepackage{setspace,graphics,rotating}
\usepackage{fullpage}
\usepackage{graphicx,color}
\usepackage{latexsym}
\usepackage{subfigure}
\usepackage{xspace}

\usepackage{algorithm}
\usepackage{algpseudocode}
\newcommand{\eat}[1]{}

\newtheorem{theorem}{Theorem}[section]
\newtheorem{lemma}{Lemma}[section]

\newtheorem{remark}{Remark}[section]
\newenvironment{proof}{\par \noindent {\bf Proof}:}{\qed \par}
\newcommand{\qed}{\nopagebreak \hfill $\Box$}

\newcommand{\beq}{\begin{equation}}
\newcommand{\eeq}{\end{equation}}
\newcommand{\baq}{\begin{eqnarray}}
\newcommand{\eaq}{\end{eqnarray}}
\newcommand{\baqm}{\begin{eqnarray*}}
\newcommand{\eaqm}{\end{eqnarray*}}
\newcommand{\barr}{\begin{array}}
\newcommand{\earr}{\end{array}}

\newcommand{\GeneSCs}{{\sc GeneSCs}\xspace}

\newif\ifabstract
\newif\iffull\fulltrue

\begin{document}

\title{Generative Models for Global Collaboration Relationships}

\author{Ertugrul Necdet Ciftcioglu}
\affiliation{IBM Research, Yorktown Heights, NY 10598}
\altaffiliation{e-mail: \{enciftci@us.ibm.com, ramanath@bbn.com, pbasu@bbn.com\}. Ciftcioglu was affiliated with Pennsylvania State University and hosted by Raytheon BBN Technologies when most of this research was conducted.}

\author{Ram Ramanathan}
\affiliation{Raytheon BBN Technologies, Cambridge, MA 02138}

\author{Prithwish Basu}
\affiliation{Raytheon BBN Technologies, Cambridge, MA 02138}

\makeatletter{\renewcommand*{\@makefnmark}{}
\date{\today}

%%%%%%%%%%%%%%%%%%%%%%%%%%
\begin{abstract}

When individuals interact with each other and meaningfully contribute toward a common goal, it results in a \emph{collaboration}, as can be seen in many walks of life such as scientific research, motion picture production, or team sports. Each individual may participate in multiple collaborations at once or over time, resulting in a non-trivial collaboration structure. The \textit{artifacts} resulting from a collaboration (e.g. papers, movies) are best captured using a hypergraph model, whereas the \textit{relation} of \textit{who has collaborated with whom} is best captured via an \textit{abstract simplicial complex} (SC). 
%Each individual may participate in multiple collaborations at once or over time---this results in a non-trivial collaboration network structure representable as an \emph{abstract simplicial complex}. Unlike a hypergraph, a simplicial complex (SC) is a set system that is \emph{closed under subset operation}; hence it captures the fundamental nature of the global structure of collaboration in groups instead of the artifacts that have been produced by such collaborations. While groups of collaborators in papers or movies can be represented by hyperedges in a hypergraph, the fundamental relationship of \emph{who has collaborated with whom} (regardless of \emph{on what}) is exactly and succinctly captured by an SC. 

In this paper, we propose a generative algorithm \GeneSCs for SCs modeling fundamental collaboration relations, primarily based on preferential attachment. The proposed network growth process favors attachment that is preferential not to an individual's {\em degree}, i.e., how many people has he/she collaborated with, but to his/her {\em facet degree}, i.e., how many maximal groups or {\em facets} has he/she collaborated within. Unlike graphs, where a node's degree can capture its first order local connectivity properties, in SCs, both facet degrees (of nodes) and \emph{facet sizes} are important to capture connectivity properties. Based on our observation that several real-world facet size distributions have significant deviation from power law---mainly due to the fact that larger facets tend to \emph{subsume} smaller ones---we adopt a data-driven approach. We seed \GeneSCs with a facet size distribution informed by collaboration network data and randomly grow the SC facet-by-facet to generate a final SC whose \emph{facet degree distribution} matches real data. We prove that the facet degree distribution yielded by \GeneSCs is power law distributed for large SCs and show that it is in agreement with real world co-authorship data. Finally, based on our intuition of collaboration formation in domains such as collaborative scientific experiments and movie production, we propose two variants of \GeneSCs based on \emph{clamped} and \emph{hybrid} preferential attachment schemes, and show that they perform well in these domains.

\end{abstract}

\maketitle

%%%%%%%%%%%%%%%%%%%%%%%%%%%

\section{Introduction}
\label{sec:intro}

\eat{
\textcolor{red}{Need to edit. Points to drive home:
\begin{itemize}
\item Is the hyperedge size distribution power law? Facet size distribution is not so, apparently. Show some plots upfront to motivate.
\item Therefore SPA type approaches \cite{HAMND2011} may not be sufficient since they are interested in hyperedge sizes and hypergraph degree
\item Also, motivate why studying maximal collaboration degree (aka facet degree) is interesting (has this been motivated in a previous paper by Minh?)
\end{itemize}
}
}

Many large endeavors in society such as scientific discoveries and production of motion pictures are a result of collaboration. Typically, individuals collaborate to form teams, for example, a scientific paper is written jointly by a team of researchers. Also, smaller teams can collaborate to form larger groups. Examples of the latter include a movie production house containing teams of artists, directors, and crew; a disaster relief mission requiring interactions between teams of medical rescue workers, fire-fighters, and law enforcement officials with some common agents serving as gateways; and a major scientific discovery happening with the coming together of research over a series of papers, which typically have some common authors. The main goal of this paper is to understand the fundamental characteristics of the underlying global collaboration structures that exist in collaborative fields such as scientific research and movie production.

In modeling collaboration structures, the basic collaborative unit could either be the \textit{relation} underlying the collaboration or the \textit{output} or \textit{artifact} from the collaboration (e.g. paper or movie). The difference in the resulting structure is best illustrated with a simple example. Suppose authors \textit{a}, \textit{b}, \textit{c}, and \textit{d} write three papers with authorships (a,b,c), (a,b), (c,d). Then a structure based on the collaboration artifact is identical to the set of papers, whereas one based on the collaboration relation is (a,b,c), (c,d). In other words, the collaboration relation structure ignores (a,b) since (a,b,c) already captures the fact that any subset of it, in particular (a,b), has collaborated. Previous studies of collaboration networks have overwhelmingly focused on the artifact-based structure~\cite{colls,PDFV2005,npacol,distcol,slov,HAMND2011,HAMND2012,Liu2012,clqc}. The relation-based structure, is instead able to capture the \textit{social} aspects of collaboration, which is interesting in its own right. 

\emph{Hypergraphs} (HG) are suitable for expressing ``richer than pairwise" relationships between collaborators~\cite{PDFV2005,HAMND2011}; however, we believe they best model the \emph{artifacts} of the collaboration -- each by a hyperedge. The collaboration \textit{relation} on the other hand is closed under the subset operation. A perfect match for succinctly capturing such a property is the \emph{abstract simplicial complexes} (SC), which in simplest terms is a collection of sets closed under the subset operation.

The primary distinction between HGs and SCs is that in the case of the latter, a ``simplex'' of dimension $k$ (modeling a $(k+1)$-ary collaboration relation) \emph{subsumes} all subset simplices of dimension $k-1$, and so on recursively. Consequently, if an HG is used to model a collaboration \textit{relation}, the distributions of sizes and degrees turn out to be non-trivially skewed compared to the relation structure itself, or equivalently, the SC representation thereof. 

In this paper, we propose a new generative algorithm \GeneSCs for SCs that models the fundamental relations underlying large-scale collaboration. \GeneSCs is primarily based on preferential attachment --- not with an individual's {\em degree} (how many people has he/she collaborated with) but rather with his/her {\em facet degree} (how many maximal groups or {\em facets} has he/she collaborated within). Unlike graphs, where a node's degree can capture its first order local connectivity properties, in SCs, there are two key metrics to consider --- a node's facet degree and a facet's size. Based upon our observation that several real-world \emph{facet size} distributions have significant deviation from power law---predominantly due to the fact that larger facets tend to \emph{subsume} smaller ones---we adopt a data-driven approach. We ``seed" \GeneSCs with a facet size distribution (from input data) and grow the SC randomly facet-by-facet to generate a final SC with a \emph{facet degree distribution} that matches real data. 

Note that we sample the facet size distribution as input instead of the artifact (hyperedge) size distribution since we want to generate the underlying SC, and not the Hypergraph, and there may be several discordant facet size distributions resulting from a single hyperedge size distribution based on how the collaboration relation is structured.

%If we were, instead, to draw from a hyperedge size distribution, the probability of subsumptions in the generated SC would be non-negligible.
%{\bf -(ENC) ACTUALLY I AM DOUBTFUL ABOUT THIS; IT MIGHT BE SLIGHTLY HIGHER COMPARED WITH USING FACETS BUT USING A GENESCS-TYPE ALGORITHM IT SEEMS EVENTUALLY THE MULTIPLE MERGE PROBABILITIES SHOULD VANISH} 
%This is because, with an SC model, we \emph{already use} the facet size distribution of the target real data set as input.{\bf-(ENC)THIS SEEMS TO BE MORE OF A DESIRED PROPERTY RATHER THAN A REASON} 

%Interestingly, while subsumption is common in real data sets of collaboration artifacts, we demonstrate that the probability of subsumptions is negligible while \GeneSCs attempts to grow a simplicial complex. This seemingly paradoxical property is desirable since we \emph{already use} the facet size distribution of the target real data set as input. If instead, the input were the hyperedge size distribution, the probability of subsumptions in the generated SC could be expected to be non-negligible.

Our key contributions are summarized below:
\begin{enumerate}
\item Systematic characterization of the nature of subsumption in real world global collaboration networks. (Section \ref{sec:subsump})
\item An efficient generative algorithm \GeneSCs to generate realistic SCs with matching facet degree distributions, given only their facet size distributions. (Section \ref{sec:genescs})
\item An analytic proof that the facet degree distribution of generated SCs is power law distributed, matching empirical studies, with an exponent $\alpha = 2 + \frac{1}{c\:s-1}$, where $c$ is the average facet density (average number of facets per node) and $s$ is the average facet size. (Theorem \ref{thm:fdeg})
\item Validation using empirical statistical analysis that the generated SCs have low Kolmogorov-Smirnov and Total-Variation distances from the real data. (Section \ref{sec:eval})
\item Adaptations of the core facet-based preferential attachment kernel in \GeneSCs to model some observed variations in real world collaboration structures. (Section \ref{sec:smoothPA})
\end{enumerate}

Interestingly, we demonstrate (and give analytical justification for the fact) that when \GeneSCs generates facets one after another with their sizes randomly drawn from the facet size distribution of the target real data set given as input, the probability of occurrence of subsumptions during this random growth process is negligible. This does not contradict our observation that subsumption phenomena is common in real collaboration artifact data. In reality, subsumptions occur over sequentially added hyperedges, whereas in GeneSCs, we already start with the pre-subsumed facet-based representation, hence further distortion is not required. This feature of GeneSCs is a valuable benefit by virtue of using facets as opposed to hyperedges in the sampling process.

%In fact, it is a desirable property of \GeneSCs, since in reality subsumptions occur over sequentially added hyperedges, whereas in \GeneSCs, we already start with the \emph{pre-subsumed} facet-based representation, hence further distortion is undesirable.

%while the subsumption phenomenon is common in real data sets of collaboration artifacts
%{\bf However, this is a desirable property of our algorithm and does not contradict with matching the real dataset, since in reality subsumptions occur over sequentially added hyperedges, whereas in our generative algorithm we already start with the \emph{pre-subsumed} facet-based representation hence further distortion is undesired.} 

%%%%%%%%%%%%%%%%%%%%%%%%%%%%%%%%%%%%%%%%%%%%%%%%%%%%%%%%%%%%
\section{Modeling Global Collaboration Relationships}

%\noindent\textbf{Modeling super-binary relations.} 
Standard \emph{graphs} are insufficient to capture group phenomena since they only model \emph{binary} relations between individuals. A generalization of graphs, namely, the \emph{hypergraph} has been proposed to address this shortcoming~\cite{PDFV2005,HAMND2011}. A hypergraph $H=(G,E)$ comprises a set of nodes $V$ and {\em hyper-edges} $E \subseteq 2^V$ to model higher order (or \emph{super-binary}) relations. 

Insights about the structure of a large collaboration network can be drawn by examining its ``artifacts", i.e., papers, movies, etc., and the underlying distributions of \emph{hyperedge size} (the number of nodes belonging to a hyperedge, e.g., number of co-authors in a paper) and \emph{hyper-degree} (the hyper-degree of a node is the number of hyperedges it belongs to, e.g., number of movies an actor has acted in). 

We believe that equally interesting insights can emerge from an understanding of the collaboration relation which focuses on the social aspects of the collaboration, namely the set of collaborating individuals.  Such a question can be answered by examining the underlying higher-order global collaboration structure, which is not concerned about the specific products of the collaboration. For example, if $a$, $b$, and $c$ have collaborated as a group $(a,b,c)$, then sparser collaboration relationships $(a,b)$ or $(b,c)$, even if they occurred, do not add much value if our goal is to understand the number of maximal groups that a person has collaborated in. 

Such information is indeed buried in the ``collaboration artifact network", i.e., the hypergraph, but typically, statistical properties of the higher-order global collaboration structure cannot be trivially determined from those of the hypergraph. In the worst case, the representation complexity of hypergraphs grows exponentially, since $k$ collaborating individuals can build as many as $2^k - 1$ different artifacts. Since we are only interested in the fact that these $k$ individuals collaborated on at least one project, the artifact network may be too unwieldy for analysis.

%%%%%%%%%%%%%%%
\begin{figure}[htbp]
  \centering
    \includegraphics[width=1.0\columnwidth]{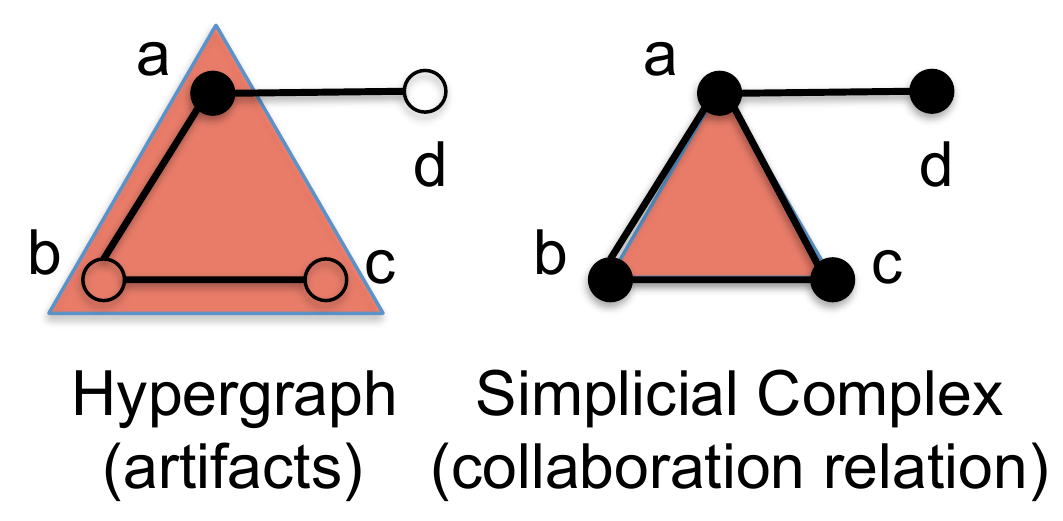}
  \caption{\label{fig:HGvsSC} Hypergraph vs. Simplicial Complex: In the hypergraph, a dark node denotes an author who has written least one paper as a sole author. In a simplicial complex, all nodes are dark since being dark just means than the corresponding author has written at least one paper, with zero or more collaborators ($0$-simplex). In this example, each $0$-simplex belongs to a $1$-simplex (edge) and in case of $a$, $b$, and $c$, also a $2$-simplex.}
\end{figure}
%%%%%%%%%%%%%%%

\subsection{Abstract Simplicial Complexes} 
The basic structure of the underlying collaboration can be modeled by an \emph{abstract simplicial complex (SC)}. A set-system $SC$ of non-empty finite subsets of a universal set $S$ is an abstract simplicial complex if for every set $X \in SC$, and every non-empty subset $Y \subset X$, $Y \in SC$ -- thus the set-system is ``closed" under subset operation~\cite{Hatcher2002}. Therefore, a simplicial complex captures the basic nature of the collaboration, i.e., who all have worked together on common tasks, instead of the artifacts that have been produced by such a collaboration, i.e., papers or movies. 

Consider the example in Figure \ref{fig:HGvsSC} -- suppose $a$, $b$, $c$, and $d$ are four authors who have co-authored five papers among them including single-author papers -- this co-authorship is denoted by a hypergraph consisting of five hyperedges: $H = \{(a,b), (a,b,c), (b,c), (a,d), (a)\}$. The collaboration structure would be represented by simplicial complex with facets $\{(a,b,c),(a,d)\}$ which {\em subsumes} the other three simplexes because if $a$,$b$, and $c$ collaborate with each other, all subsets of them do so as well. While $a$ and $c$ have not explicitly collaborated {\em separately}, they have collaborated with each other in presence of $b$ in the paper denoted by $(a,b,c)$. Essentially, a simplicial complex consists of a set of maximal simplexes or {\em facets}. 

In previous work, we have shown how simplicial complexes can be effectively used to model collaboration networks~\cite{Ramanathan2011,HoangRMS13,HoangRS14}. For a collaborative group denoted by a facet, two basic metrics are \emph{facet size} (how many people belong to that collaboration) and a node's \emph{facet degree} (how many maximal collaborations or facets does that node belong to).

\subsection{Modeling Subsumptions}
\label{sec:subsump}

The basic difference between hypergraphs and simplicial complexes can be explained by the phenomenon of \emph{subsumptions}. In the previous example, hyperedges $(a,b)$ and $(b,c)$ get \emph{subsumed} by the largest hyperedge $(a,b,c)$, which is a \emph{facet} in SC. Similarly, facet $(a,d)$ subsumes $(a)$.

In theory, subsumptions can be very pronounced. Consider a large research project with $N$ participating faculty members. Consider the situation where each faculty member has one single-author paper, one paper written with one of the other faculty members, one paper written with two \eat{other} distinct faculty, and so on. Finally, assume that all of these $N$ authors collaborate to write a joint paper together. Clearly, there are $2^N-1$ distinct hyperedges in total -- $N$ single author papers, $\binom{N}{2}$ two-author papers, and $\binom{N}{k}$ $k$-author papers, in general. Since each node has exactly $2^{N-1}$ hyperedges, the hyperedge degree distribution is given by the impulse function $\delta(2^{N-1})$. On the other hand, hyperedge size distribution is a non-monotonic function which is proportional to $\frac{N!}{(N-k)!k!}$, centered around $\frac{N}{2}$. In contrast, in the simplicial complex representation, the largest collaboration is the only \emph{facet}, so all nodes have facet degree $1$, hence the facet degree distribution is $\delta(1)$. Moreover, the only facet has size $N$, hence the facet size distribution is $\delta(N)$ (See Fig. \ref{fig:extremeSubs}). There is significant discrepancy between the two distributions due to the intense degree of subsumption -- since all faculty members collaborate with each other on one paper, the other smaller collaborations are directly implied by the former. Note that the deviation would be larger if there were multiple papers with exactly the same authors, since the hyperedge count would increase without affecting facet statistics.

%%%%%%%%%%%%%%%%
\begin{figure}[htbp]
\centering \centering
\vspace{-0.15in}
\includegraphics[width=1.0\columnwidth, ]{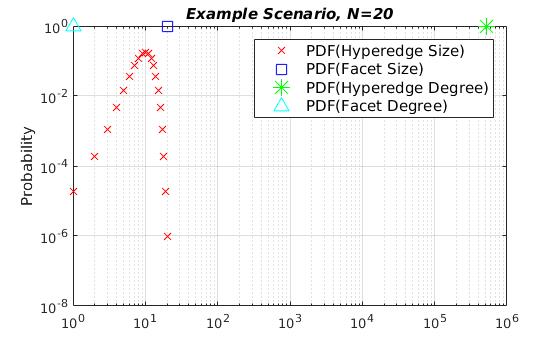}\vspace{-0.2in}
\caption{\label{fig:extremeSubs} A hypothetical case with extreme subsumptions}\vspace{-0.1in}
\end{figure}
%%%%%%%%%%%%%%%%

The above example demonstrated a case where many different hypergraph instances associated with $N$ nodes may map to only one simplicial complex, i.e. a many-to-one mapping between the set of different hypergraphs to one simplicial complex representation. Yet, one may think that once a specific hypergraph is given, the corresponding simplicial complex can be simply obtained from it by following the subset closure operation. However, in reality, while working with datasets, one typically expects only the distributional statistics to be given. We next demonstrate that starting from a given hyperedge size distribution one might end up with multiple simplicial complex representations with drastically different facet size distributions even for a fixed number of nodes. In fact, the total variation distance $D_{TV}$ (a commonly used distance metric to compare two different distributions, i.e., normalized statistical distance between two distributions which measures sum of differences over the support set, formally defined in Section \ref{sec:eval}) among the various feasible non-isomorphic simplicial complexes might approach $1$, which is the maximum value that can be defined between two distributions, as $N\rightarrow \infty$.

Consider a given hyperedge size distribution $f_h(.)$. Let us denote the maximum hyperedge size in $f_h(.)$ by $H$, and assume that $N=2H-1$, where again $N=|V|$ denotes the number of nodes.
\eat{Next,  let the hyperedge size distribution be such that there are $2H-1$ size one hyperedges, one size $H-1$ hyperedge and one size $H$ hyperedge on an average when one samples $2H+1$ hyperedges;} 
Next, consider a hypergraph consisting of a total of $2H+1$ hyperedges, with $2H-1$ hyperedges of size one, and one hyperedge each of size $H-1$ and $H$. That is, $f_h(1)=\frac{2H-1}{2H+1}, f_h(H-1)=\frac{1}{2H+1}$ and $f_h(H)=\frac{1}{2H+1}$. A small scale illustration of this scenario with $N=5$ and $H=3$ is given in Figure \ref{fig:HGmulSC}.
Given this distribution, it may be possible that the two large hyperedges are disjoint and do not possess any nodes in common, hence in the SC representation there is one facet with size $H-1$ and one facet with size $H$, together spanning all $N=2H-1$ of the nodes. Accordingly, all the smaller (single node) hyperedges are subsumed by the two larger facets since every node belongs to a larger collaboration (e.g. $SC_1$ in Figure \ref{fig:HGmulSC}). On the other hand, it could also be the case that the larger hyperedge of size $H$ subsumes the one of size $H-1$. Then, $H-1$ of the $2H-1$ hyperedges of size one do not belong to the larger facets, and hence are disjoint. Overall, there are $H-1$ size one facets and one size $H$ facet in the simplicial complex representation. (e.g. $SC_2$ in Figure \ref{fig:HGmulSC}).
%A  ----{\bf I had some ambiguity in histogram vs distribution while considering the hypergraph; should we start from a collection of hyperedges and define its distribution, or start with the distribution and talk about a generated realization, perhaps the first one is clearer to make the case so I went with that}-----

%%%%%%%%%%%%%%%
\begin{figure}[htbp]
  \centering \vspace{-0.2in}
    \includegraphics[width=1.0\columnwidth]{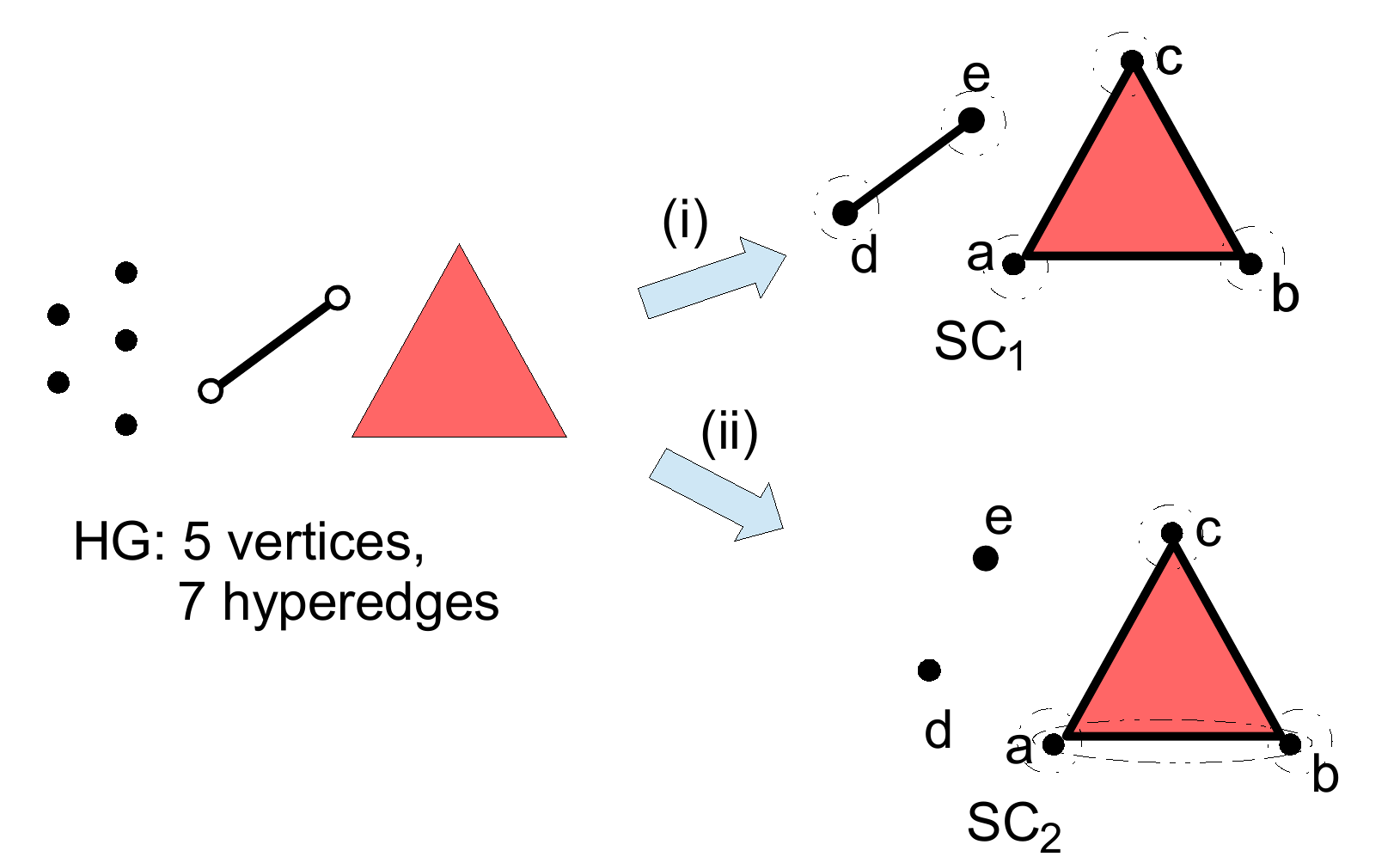}\vspace{-0.15in}
  \caption{\label{fig:HGmulSC} Two of the possible different Simplicial Complexes and the corresponding facet size distributions $z(s)$, resulting from a given Hyperedge size distribution $f_h(s)\approx(0.714, 0.143, 0.143)$ with $N=5$ nodes: (i) $SC_1 = \{(a,b,c), (d,e)\}$, $z_{SC_1}(s)=(0,0.5,0.5)$ (ii) $SC_2 = \{(a,b,c), (d), (e)\}$, $z_{SC_2}(s)=(0.\bar{6}, 0,0.\bar{3})$. Total variation distance ($D_{TV}$) between the two facet size distributions: $0.\bar{6}$ The dashed contours depict hyperedges which were subsumed in the facet representations.}\vspace{-0.1in}
\end{figure}
%%%%%%%%%%%%%%%
%%%%%%%%%%%%%%%
\begin{figure*}[htbp]
  \centering
    \includegraphics[width=0.48\textwidth]{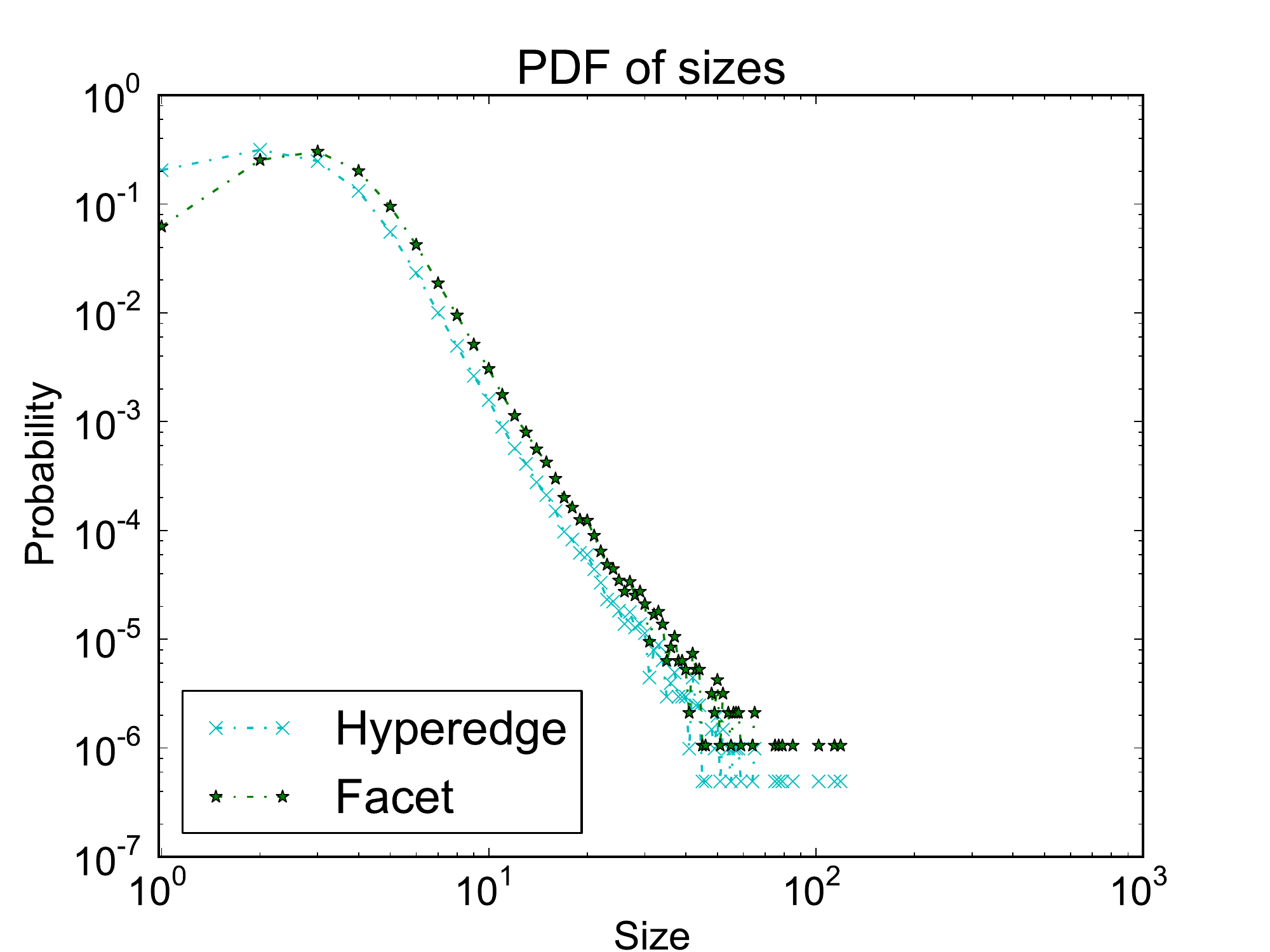}
    \includegraphics[width=0.48\textwidth]{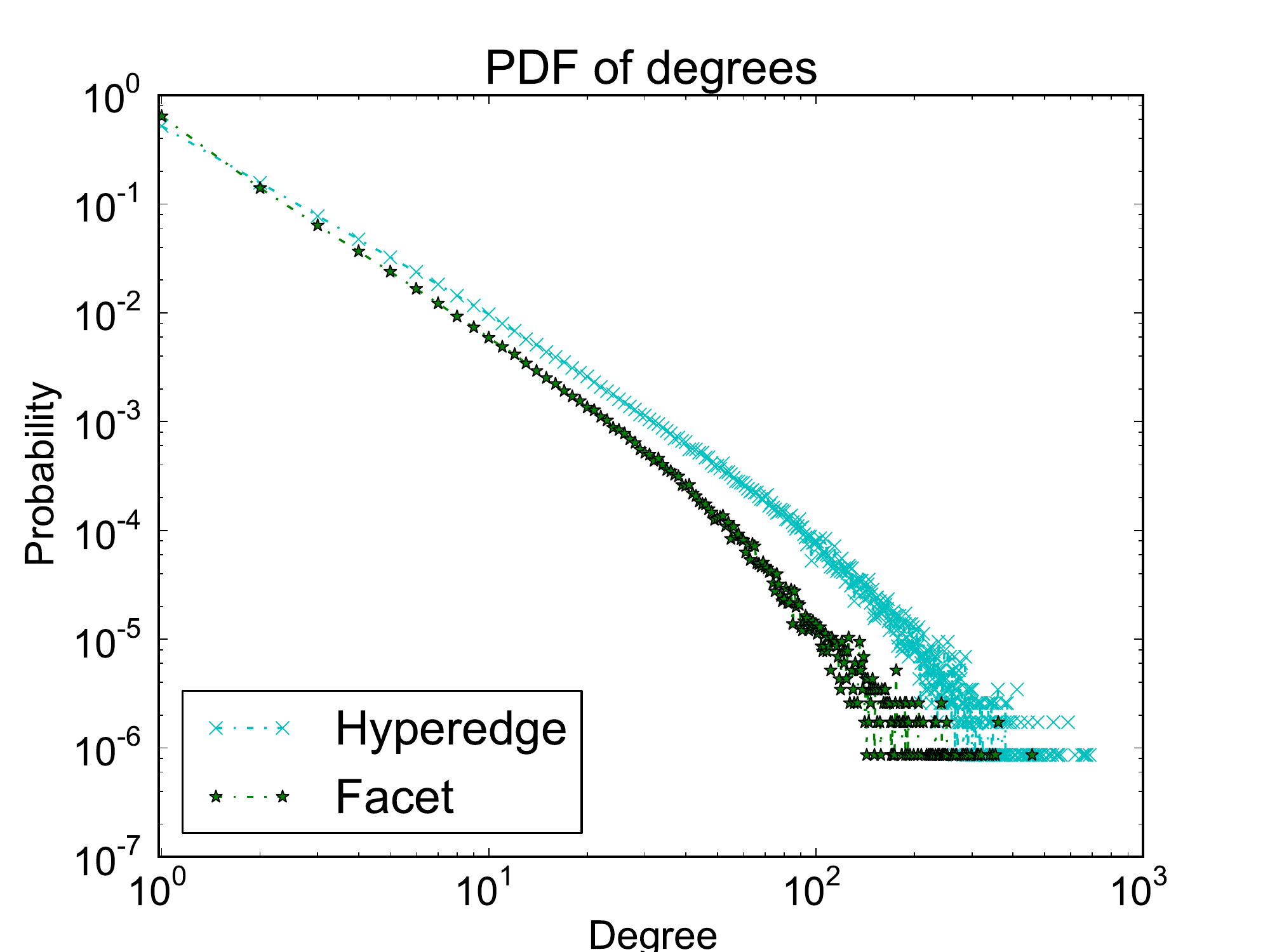}
  \caption{\label{fig:DBLP-HGvsSC} Hypergraph vs. Simplicial Complex metrics for DBLP: (a) size and (b) degree}
\end{figure*}
%%%%%%%%%%%%%%%
\begin{figure*}[t!]
  \centering
    \includegraphics[width=0.48\textwidth]{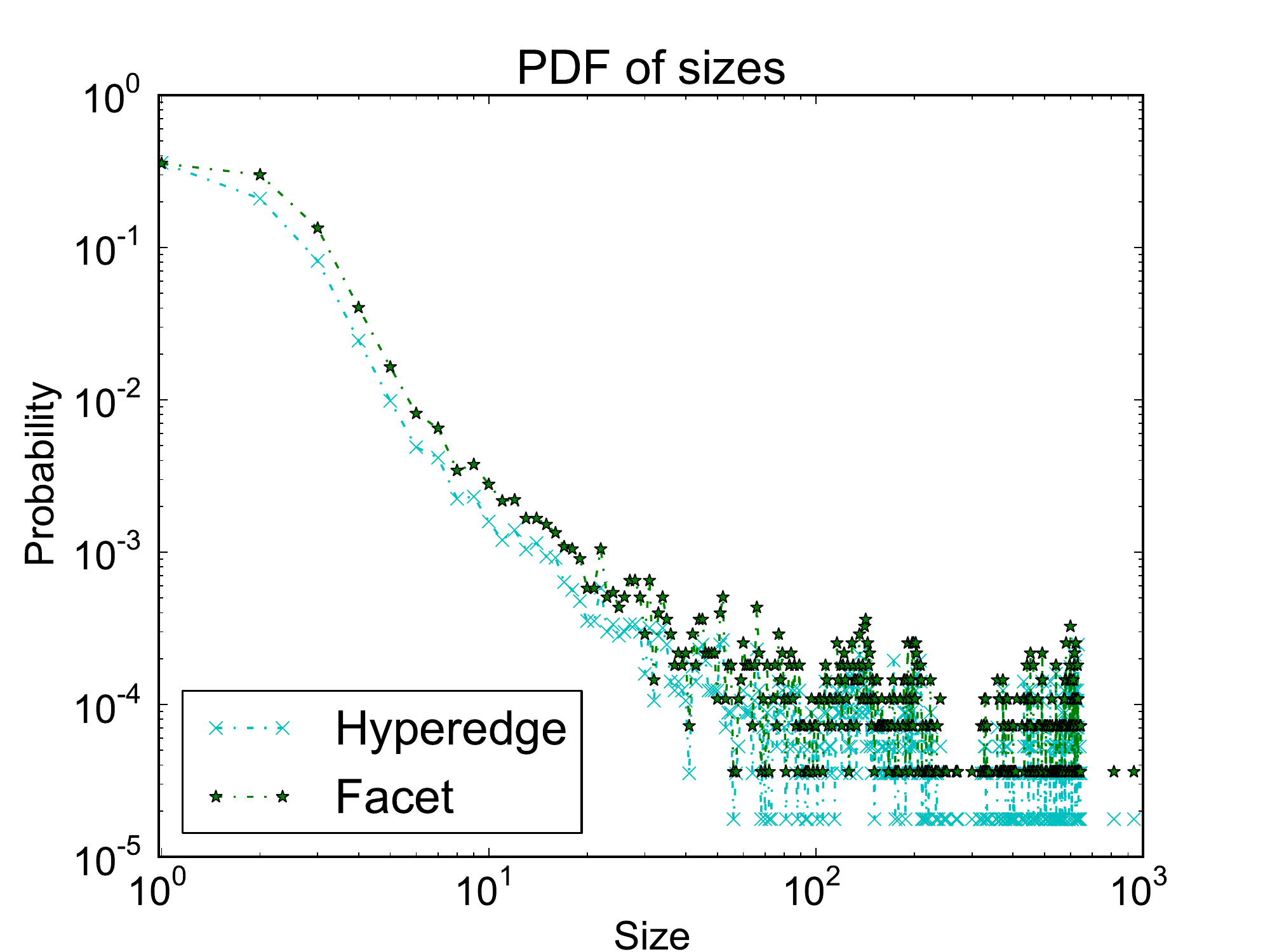}
    \includegraphics[width=0.48\textwidth]{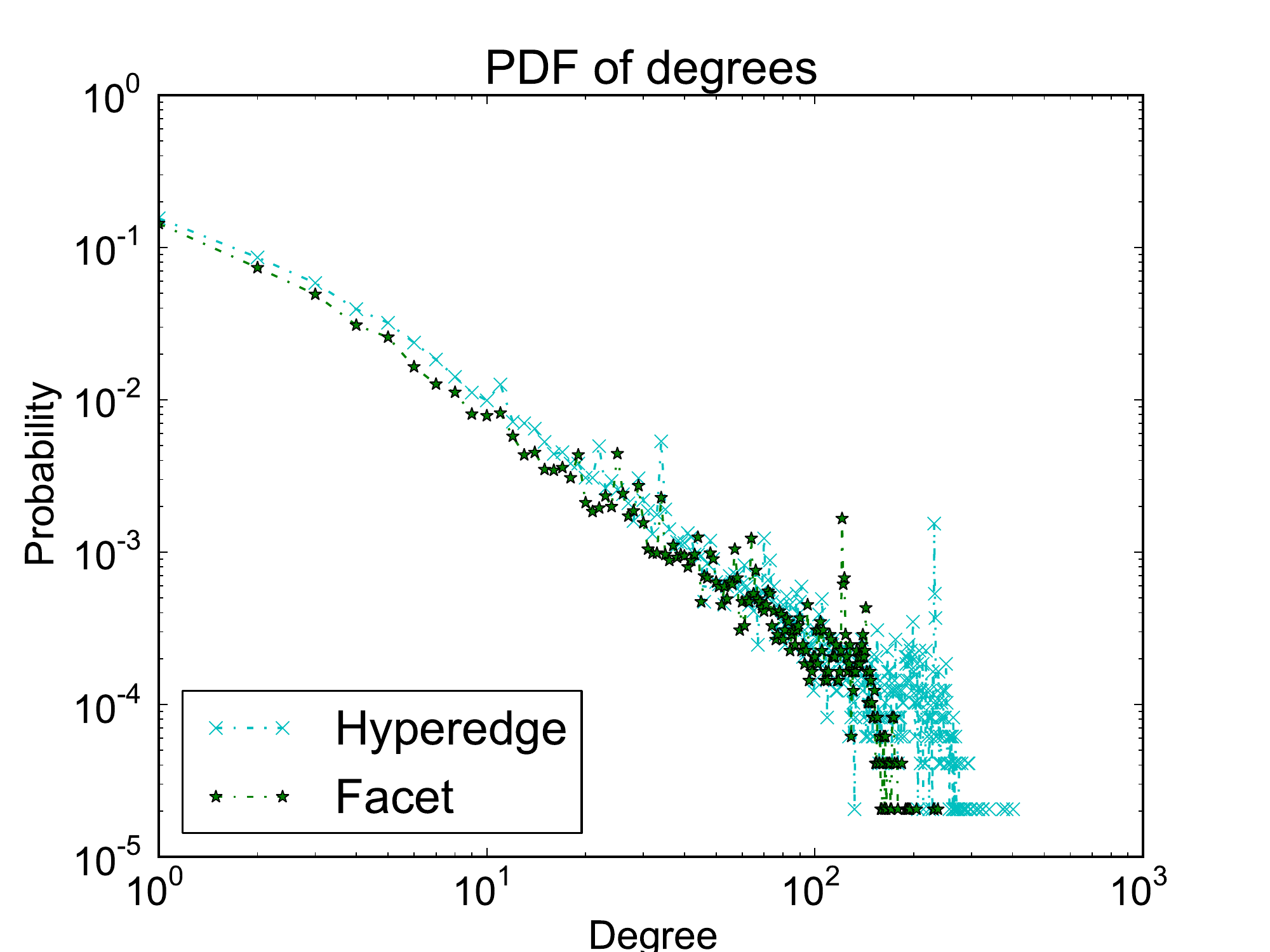}
  \caption{\label{fig:PRD-HGvsSC}Hypergraph vs. Simplicial Complex metrics for Physical Review D: (a) size and (b) degree}
\end{figure*}
%%%%%%%%%%%%%%%

It can be observed from Figure \ref{fig:HGmulSC} that the facet size distribution of $SC_1$ is an $H$-dimensional vector $(0,\ldots,0.5,0.5)$ with non-zero entries at $H-1$ and $H$; and the facet size distribution of $SC_2$ is $(\frac{H-1}{H},0,\ldots,\frac{1}{H})$ with $f_S(H)=\frac{1}{H}$. The total variations (please see Section \ref{sec:eval} for a formal definition) between two distributions are $0.5(\frac{H-1}{H}+ 0.5+ |0.5-\frac{1}{H}|)$, which converges to $1$ as $N$ (and hence $H$) grows. Even if the condition $N=2H-1$, i.e., $H\approx N/2$ may be found to be restricting in the sense that $H$ may not grow as much, alternative examples and expressions can be constructed.

For example, with $N= 3H-2$, where the maximum hyperedge is of size $H\approx N/3$, assume we have a hypergraph on $N$ nodes which consists of one hyperedge of size $H$, two hyperedges with size $H-1$, and $3H-2$ with size one hyperedge. This corresponds to the hyperedge size distribution $f_h(1)=\frac{3H-2}{3H+1}, f_h(H-1)=\frac{2}{3H+1}$ and $f_h(H)=\frac{1}{3H+1}$. Then, if all large facets are disjoint they cover all nodes and one has a facet size distribution of $(0,0,....0.\bar{6},0.\bar{3})$, whereas if the two facets of size $H-1$ differ by only one node and are both subsumed by the one of size $H$; $2H-2$ of the singleton nodes remain disjoint, and the facet size distribution for this scenario is $(\frac{2H-2}{2H-1},0,...,\frac{1}{2H-1})$ and the total variation would still approach $1$ for large enough $H$.

The above examples clearly demonstrate that if one in interested in understanding global collaboration relationship structures represented as simplicial complexes, starting from hyperedge size statistics to generate SCs may result in wildly discordant structures. Hence, in our generative algorithm \GeneSCs (in Section \ref{sec:gengrowth}) we take the facet size distribution as input, instead.

%%%%%%%%%%%%%%%
\begin{figure*}[t!]
\centering
\includegraphics[width=0.32\textwidth]{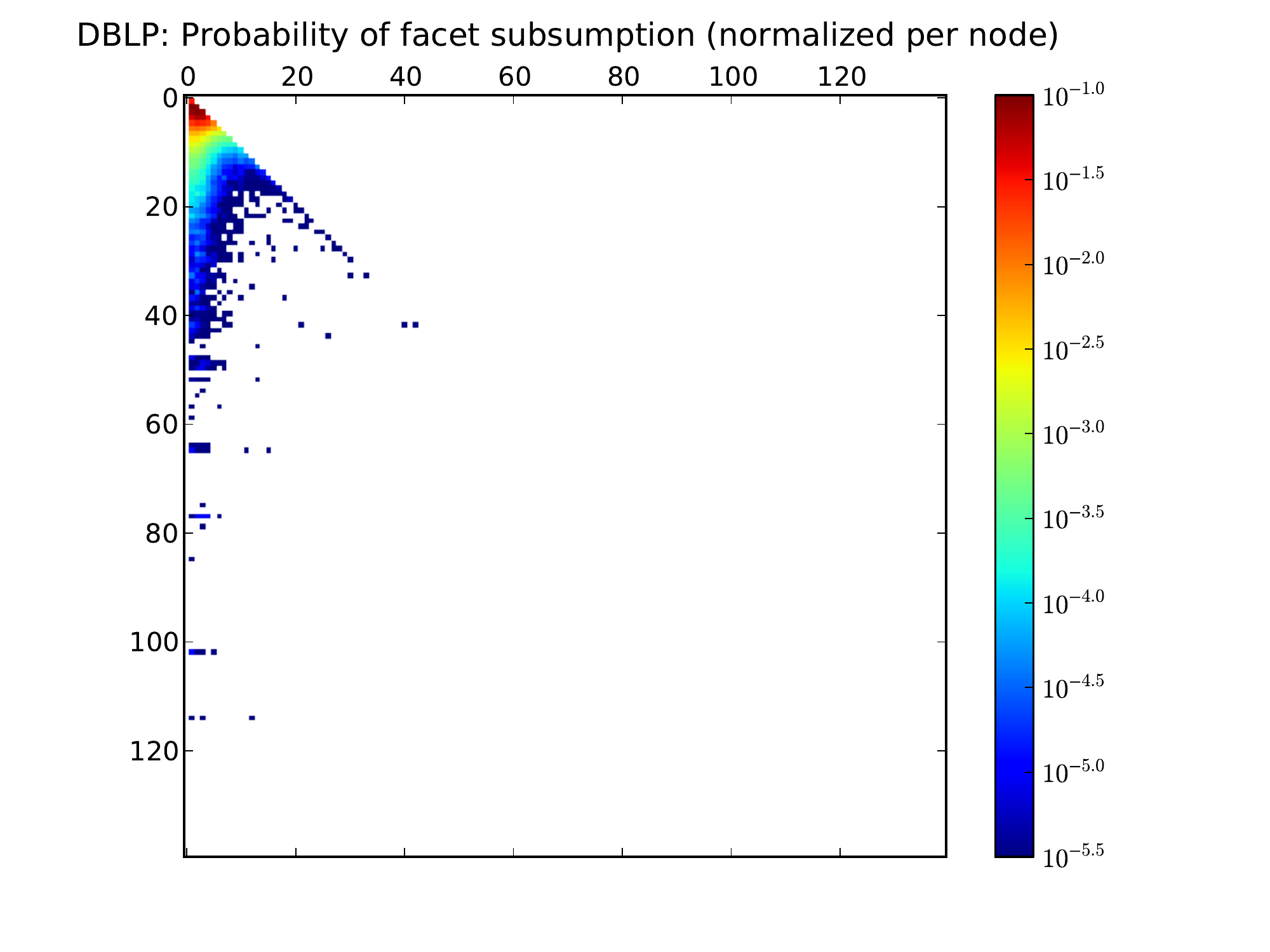}
\includegraphics[width=0.32\textwidth]{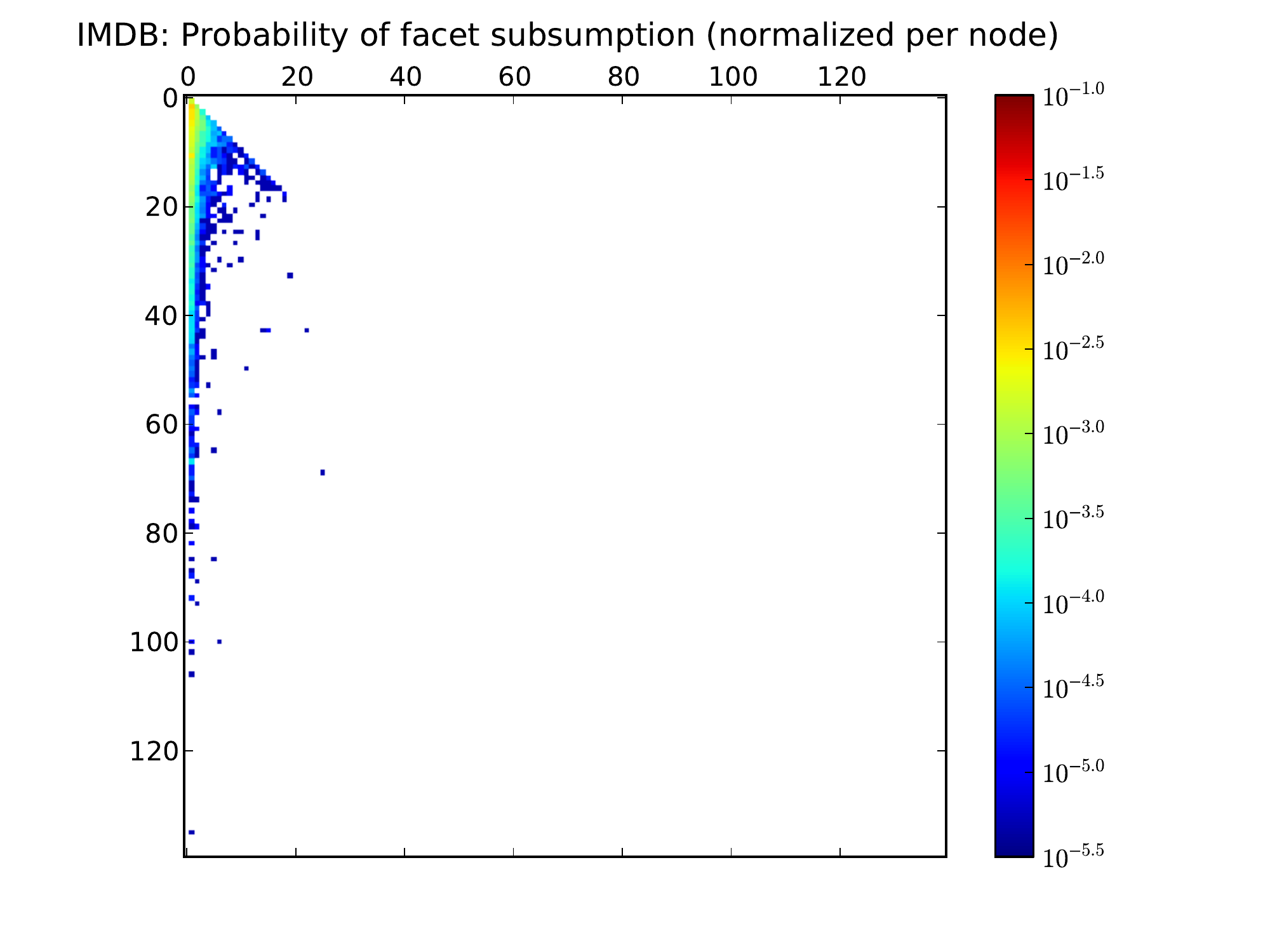}
\includegraphics[width=0.32\textwidth]{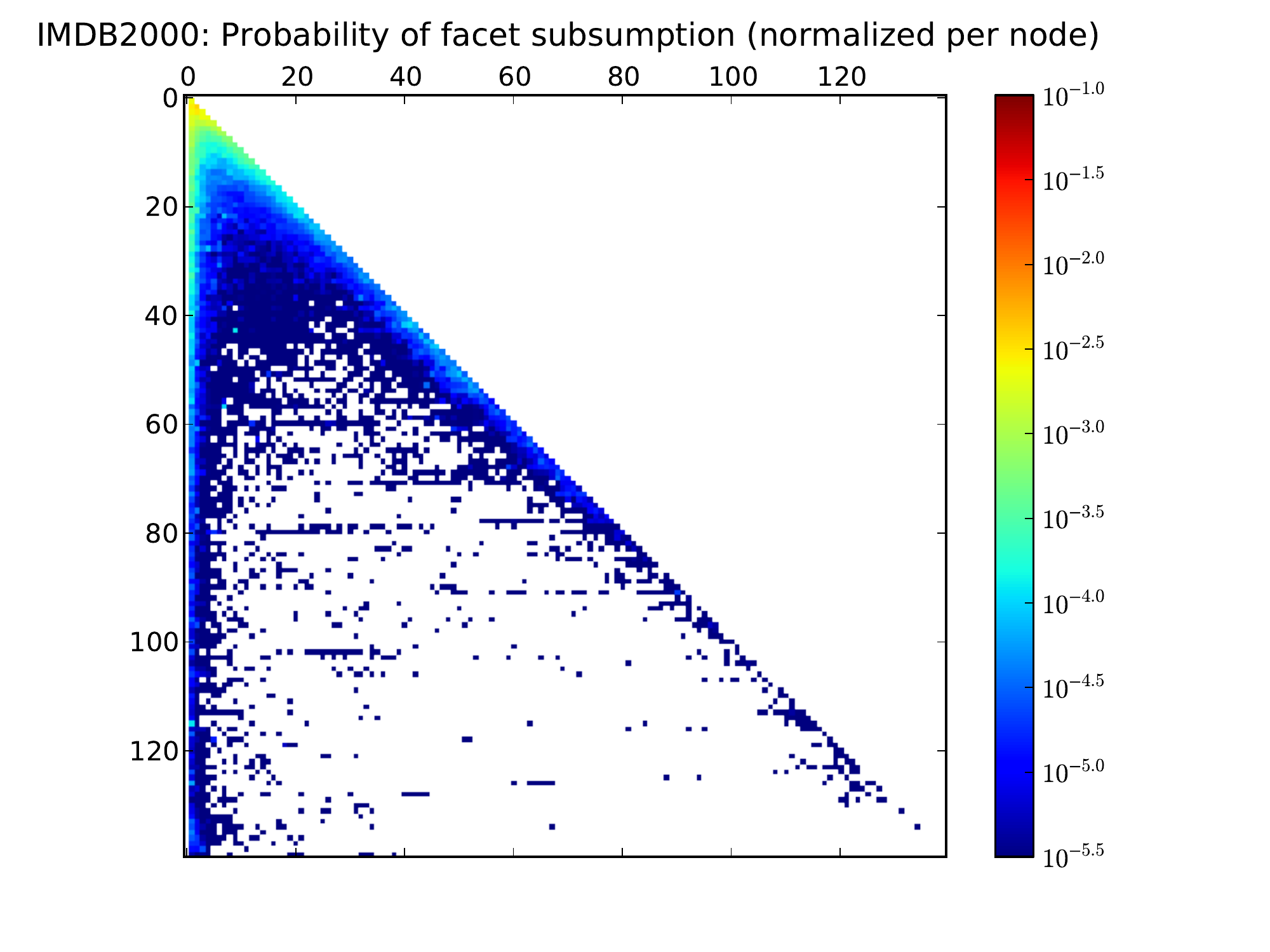}
\includegraphics[width=0.32\textwidth]{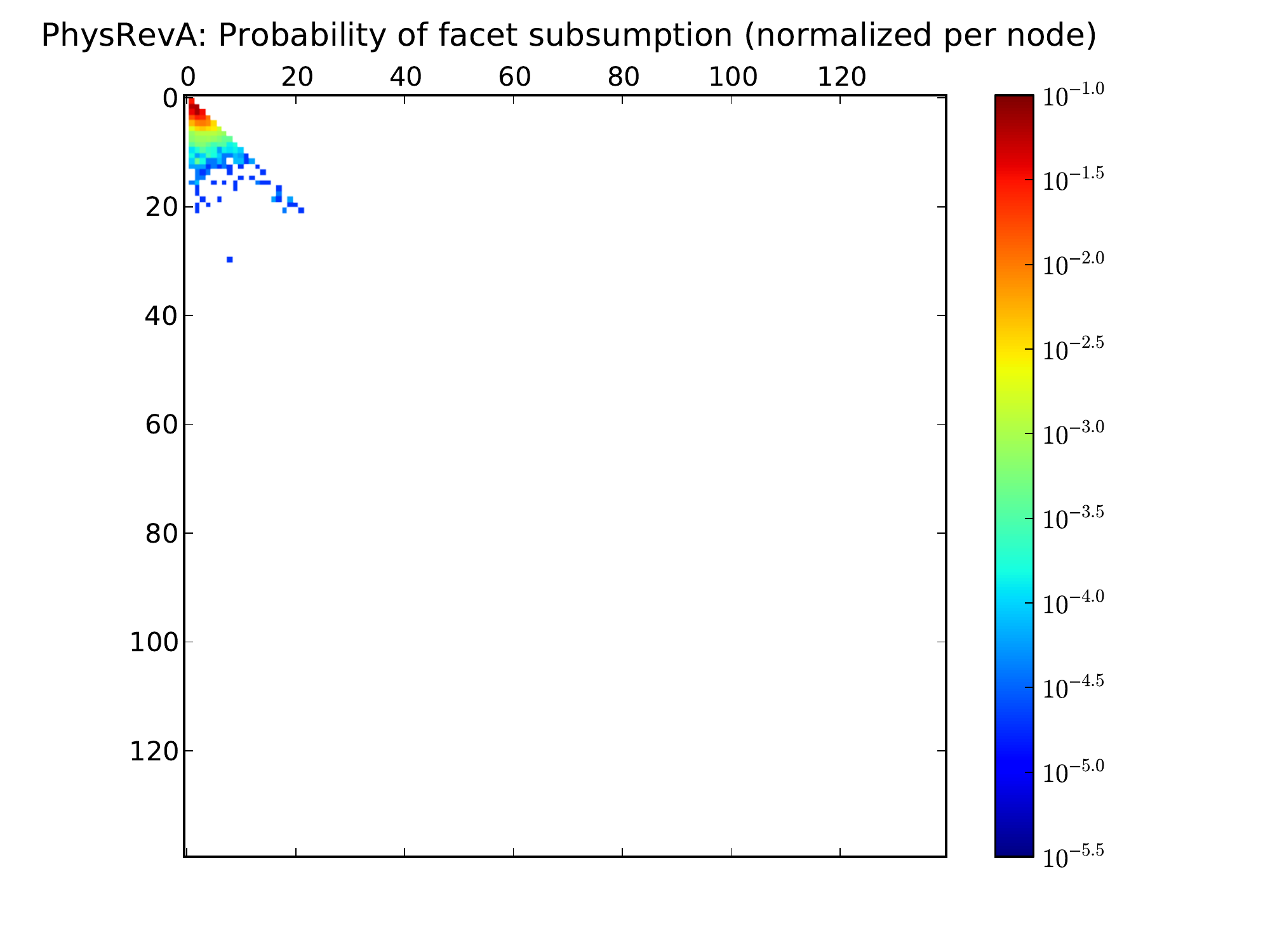}
\includegraphics[width=0.32\textwidth]{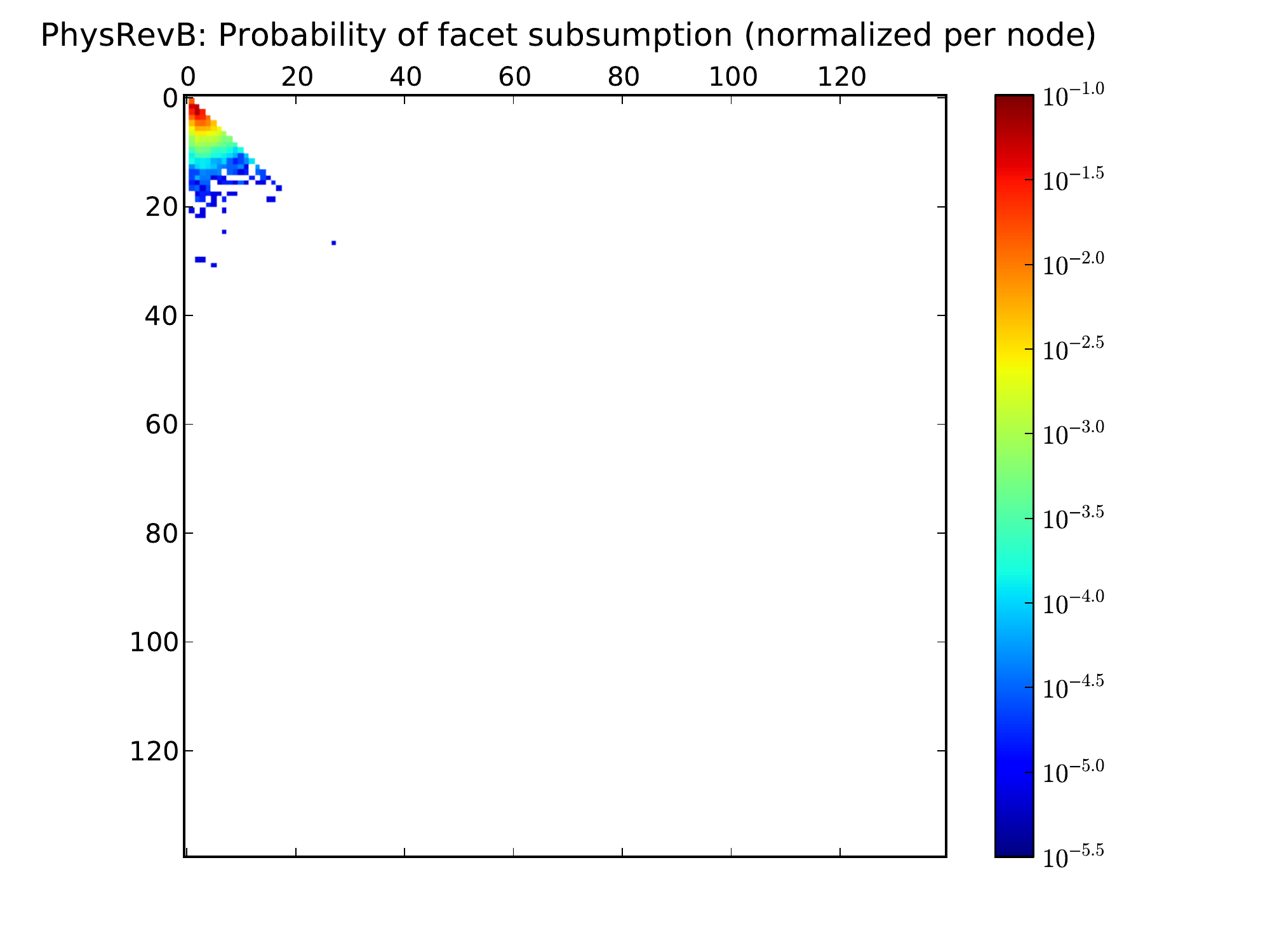}
\includegraphics[width=0.32\textwidth]{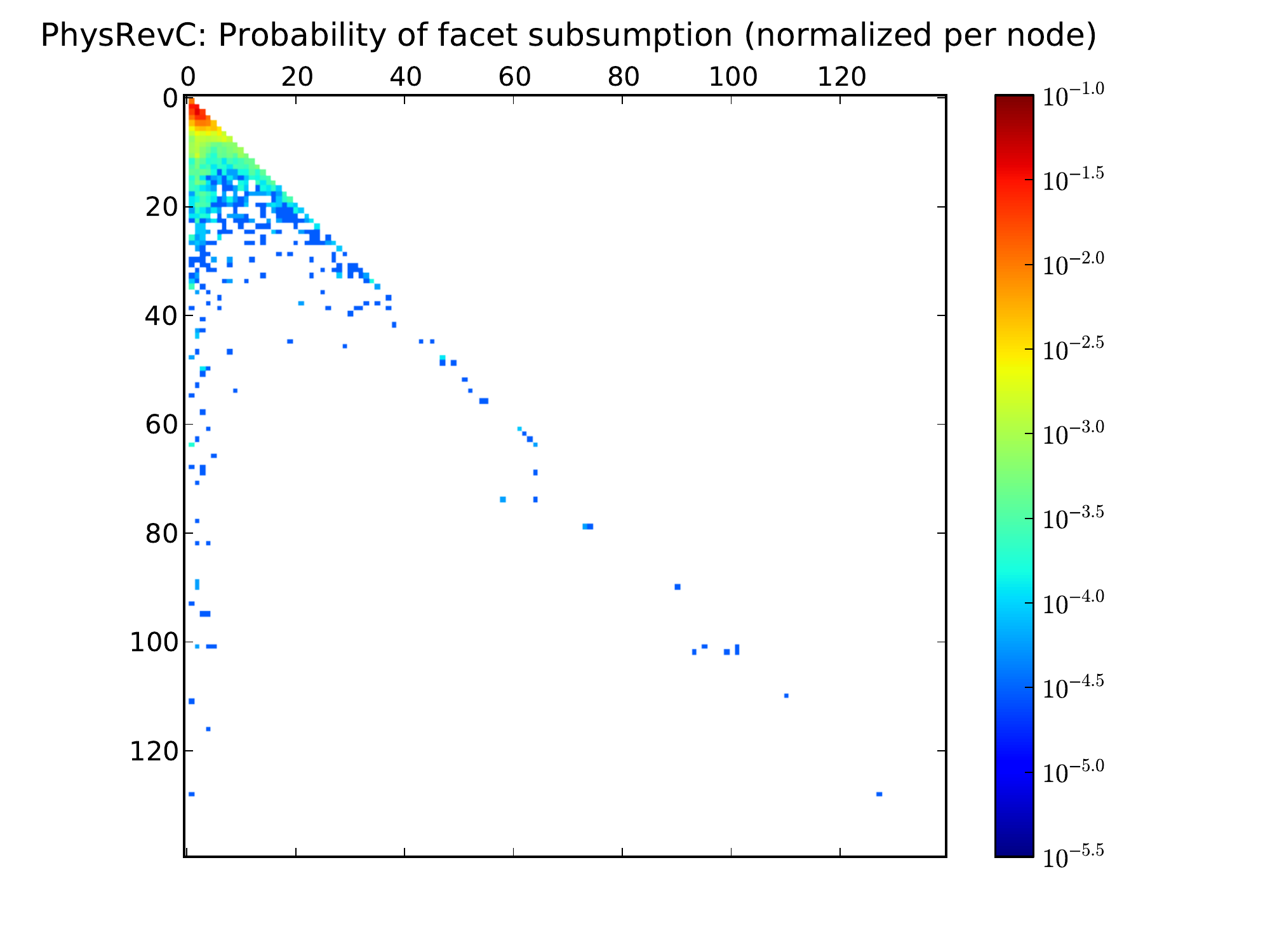}
\includegraphics[width=0.32\textwidth]{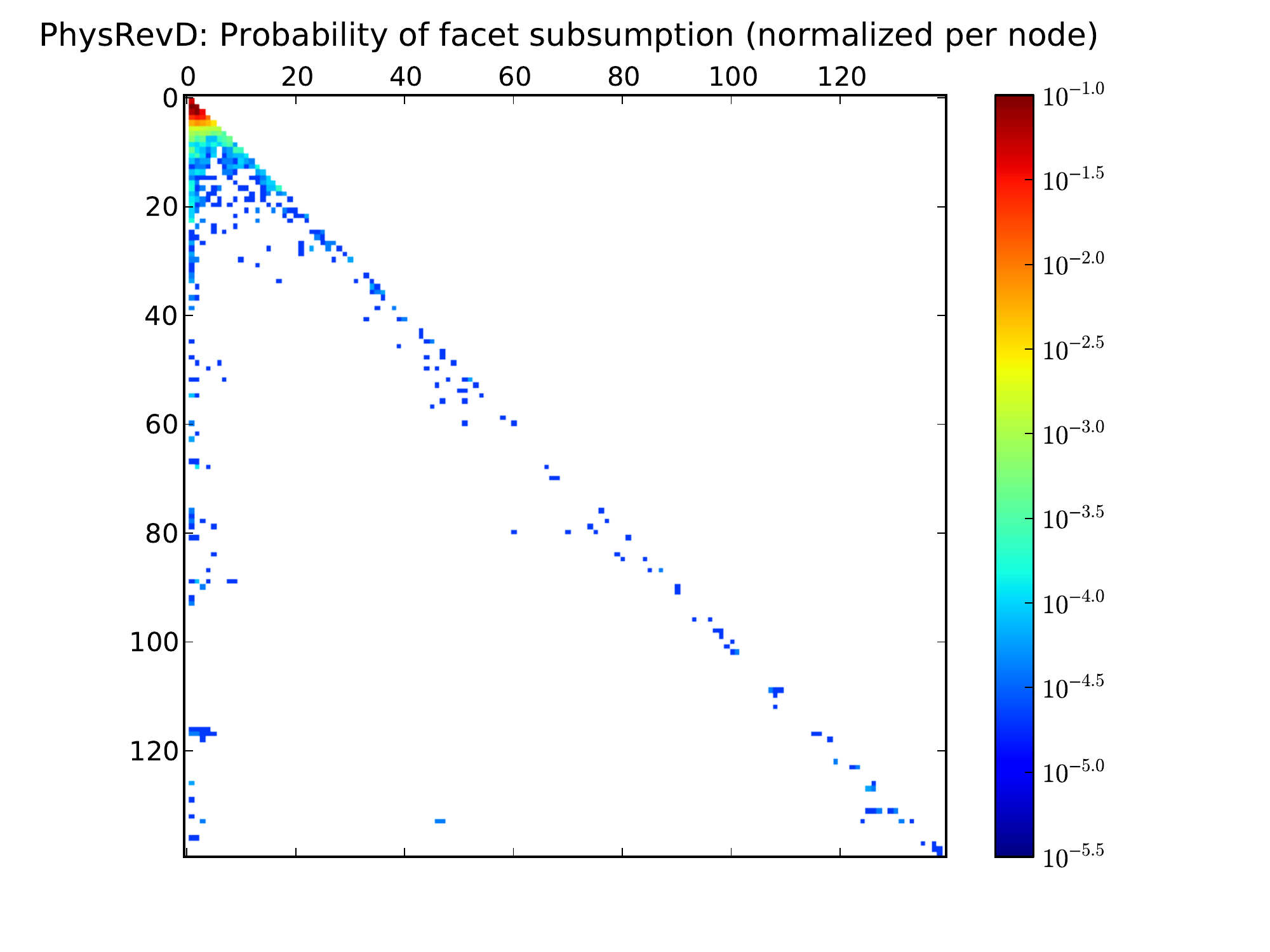}
\includegraphics[width=0.32\textwidth]{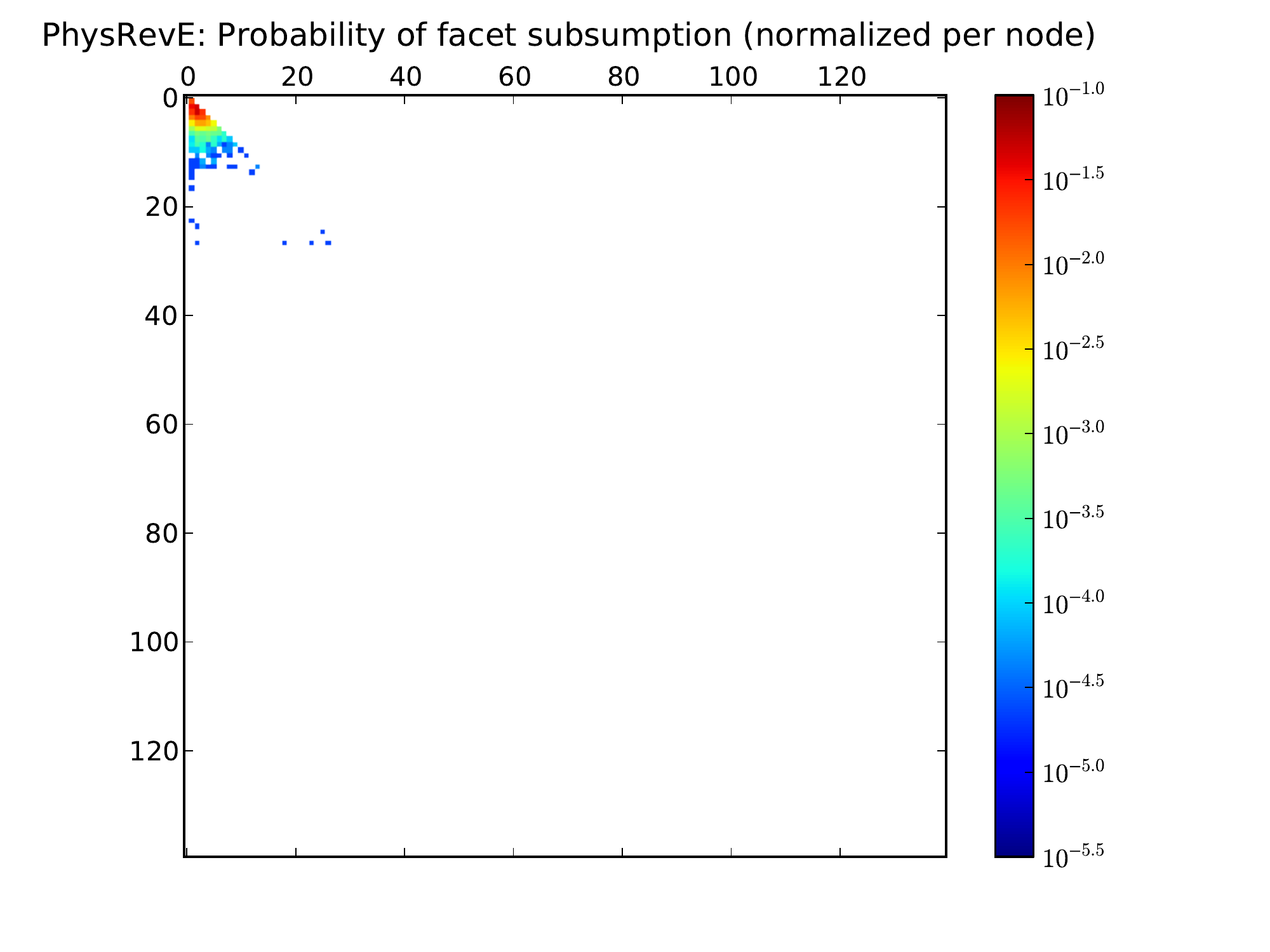}
\includegraphics[width=0.32\textwidth]{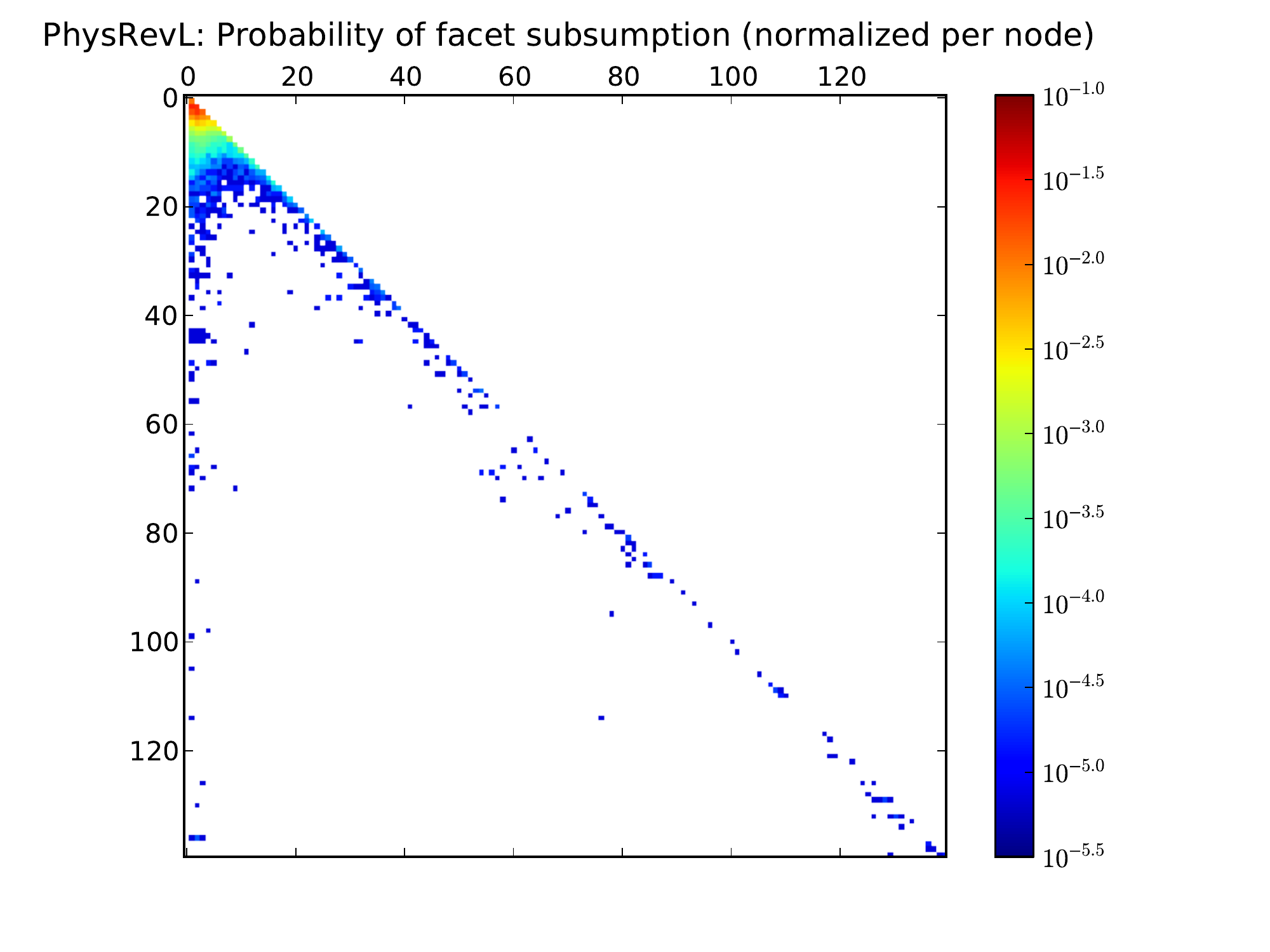}
\caption{\label{fig:subsumptions} Subsumptions in collaboration network data (in each plot, rows and columns correspond to sizes of \emph{subsuming} and \emph{subsumed} facets, respectively): (a) DBLP, (b) IMDB (only regular movies and cast), (c) IMDB 2000 (all movies, cast, and crew year 2000 onwards), (d) Phys Rev A, (e) Phys Rev B, (f) Phys Rev C, (g) Phys Rev D, (h) Phys Rev E, (i) Phys Rev Letters}\vspace{-0.2in}
\end{figure*}
%%%%%%%%%%%%%%

While the above examples can be regarded as rather unlikely events, subsumption does occur quite frequently in real world collaboration. Figure \ref{fig:DBLP-HGvsSC} illustrates statistics for both hypergraph and SC models of DBLP co-authorship data. While the tail of the facet size distribution obeys a \emph{power law}, the head, where significant probability mass is concentrated, does not. In particular, a significant number of singleton authors get subsumed by the larger collaborations they participate in. Additionally, the slopes of the facet degree and hyperedge degree distributions are different. This can also be attributed to \emph{subsumptions} of smaller hyperedges by larger facets. Figure \ref{fig:PRD-HGvsSC} illustrates the difference between hyperedge and facet statistics for another co-authorship data set (Physical Review D journal). Like in Figure \ref{fig:DBLP-HGvsSC}, there is significant difference at the head of the size distributions, implying a significant frequency of subsumptions. Also, in both Figures \ref{fig:DBLP-HGvsSC} and \ref{fig:PRD-HGvsSC}, the tails of the facet degree distribution are shorter than those of the respective hyper-degree distributions. This can be attributed to the subsumption of many small hyperedges at the high-degree nodes, thus reducing the facet degree compared to their hyperedge degrees.

%\noindent\textbf{Subsumptions in real data sets.}  
We now examine the nature of subsumptions in real collaboration datasets more closely. In Figure \ref{fig:subsumptions}, we plot nine real collaboration data sets, the number/percentage of subsumptions of facets of size $i$ by facets of size $j$ (where $j \geq i$). Obviously, this is a lower triangular matrix where the rows indicate sizes of \emph{subsuming} facets and columns indicate sizes of \emph{subsumed} facets.

We observe that the nature of subsumptions varies across the nine data sets. In DBLP, IMDB (with regular movies and cast only), Phys Rev A, Phys Rev B, and Phys Rev E, small facets are subsumed (first few columns). This is intuitively expected, since it is likely that smaller subsets of a collaboration are also valid collaborations. 

In Phys Rev D, Phys Rev L and Phys Rev C, there is a strong subsumption presence on and off the diagonals. These subsumption events model scenarios where a significant number of individuals are collaborating on a task and the exact set of individuals collaborate again, with perhaps a few additional collaborators such as a new graduate student joining a lab. These situations likely arise from the large endeavors typical of experimental physics where very large collaborations of laboratories result in a paper, as evidenced by Phys Rev D and L. In fact, for Phys Rev D, the diagonal has non-trivial mass even for collaboration sizes of 500, which have not been shown here. 

Finally, the IMDB data set that includes both cast and crew of movies (c) is a class apart in the above trend taken to a much larger magnitude. This is because a core crew tends to get utilized by a director in multiple movies. As seen from these figures, such events occur often in reality, hence subsumption is an important issue to address when modeling the structure of the underlying global collaboration network.

%We observe two types of subsumptions: cases where small facets are subsumed (first few columns), and cases along the diagonal and off the diagonal. The former is intuitively expected, since it is likely that smaller subsets of a collaboration of any size are also valid collaborations. The latter pattern is more counter-intuitive -- these subsumption events model scenarios where a significant number of individuals are collaborating on a task and the exact set of individuals collaborate again, with perhaps a few additional collaborators such as a new graduate student joining a lab. These situations are common particularly in large endeavors such as experimental physics where very large collaborations of laboratories result in a paper. This is evidenced in Phys Rev D (g) and Phys Rev L data (i); in fact, for Phys Rev D, the diagonal has non-trivial mass even for collaboration sizes of 500, which have not been shown here. An even more pronounced trend is observed in the IMDB data set that includes both cast and crew of movies (c). This is because a core crew tends to get utilized by a director in multiple movies. As seen from these figures, such events occur often in reality, hence subsumption is an important issue to address when modeling the structure of the underlying global collaboration network. 

%%%%%%%%%%%%%%%%%%%%%%%%%%

\section{Generative Growth models for Networks
Induced by Global Collaboration Relationships}
\label{sec:gengrowth}

Generative network growth models have received great interest over the past 15 years for classical (binary) graphs. The most prominent growth model has been \emph{preferential attachment}, which has been demonstrated to result in graphs with node degrees following power law distributions~\cite{DMS2000}, i.e. $k^{-\gamma}$, where $\gamma \in [2,3]$. Examples of other network growth models that have received attention are small-world models~\cite{Watts1998}, densification models~\cite{Leskovec2007}, and duplication models~\cite{Chung2003}, to name some.

In contrast, there has been a limited amount of work on generative models for group collaboration structures. In addition to the \emph{node degree} distribution, which has been a focal metric for classical network growth models, for collaboration structures more complex than graphs, the hyperedge size distribution is key. Hebert-Dufresne et al.~\cite{HAMND2011} proposed a generative algorithm called Structural Preferential Attachment (SPA) by progressively growing a hypergraph based on parameters that depend on the power law exponents of hyperedge size and degree distributions (their assumption was that both obey power law distribution). In SPA, \emph{two} free probability parameters are simultaneously controlled to generate new structures and attachment points in the current network in order to simultaneously match the tails of both the hyperedge degree and size distributions. However, as Figures \ref{fig:DBLP-HGvsSC} and \ref{fig:PRD-HGvsSC} suggest, this is not accurate for many collaboration \textit{relation} structures, particularly for the size distributions. More importantly, {\em structure-based} growth models~\cite{PDFV2005,HAMND2011} do not focus on modeling the structure of the underlying global collaboration network -- instead they model the ``artifacts" of collaboration, i.e., hyperedges.

%hyperedges -- these are often observed to be power-law distributed in size. We however observe that the sizes of maximal collaboration (or facets) are often far from being power-law distributed in reality, even though the facet degree distribution {\em is} distributed as power-law or as power-law with exponential cutoff. 

To generate real world collaboration relation structures, we propose a generative facet-by-facet growth model based on preferential attachment, not based on an individual $X$'s {\em degree}, that is, how many people has $X$ collaborated with over a period of time, but with $X$'s {\em facet degree} which measures how many maximal collaborations or {\em facets} has $X$ been involved in. The basic intuition behind this is the following: individuals who are comfortable being part of several distinct collaboration endeavors are likely to attract more collaborators than the individuals who are happy participating in a fewer number of collaboration endeavors (albeit with several collaborators).

We assume that the facet size distribution is given to us and our primary aim is to generate a random simplicial complex that closely matches the ground truth distribution for facet degrees. We do this because the facet size distributions are often \emph{non-power law} and even \emph{non-monotonic}, particularly near the head of the distribution as in Figures \ref{fig:DBLP-HGvsSC} and \ref{fig:PRD-HGvsSC}, thus not obeying trends of typical distributions generated by random growth models. On the other hand, facet degree distribution in Figure \ref{fig:DBLP-HGvsSC} has monotonic behavior and is power law with exponential cutoff, implying that it might be possible to obtain it more accurately via random growth models. Accordingly, our facet-based generative model takes as input the facet size distribution and grows the \emph{collaboration simplicial complex} one facet at a time. During this process, a large facet may {\em subsume} a smaller existing facet since we are interested in capturing the fundamental structure of collaboration. 

In related work, we first acknowledge the early work in~\cite{colls} which empirically points out the basic  properties of scientific collaboration network structures. Over the past decade, a significant number of researchers have studied structures representing collaboration artifacts. Liu et al.~\cite{Liu2012} have proposed a preferential-attachment based growth model for ``affiliation networks". However, this is only a qualitative study without any theoretical analysis. The artifacts of collaboration are addressed in~\cite{npacol}, which again provides analytic results for only very special cases where the collaborative outputs are fixed size. This was also the assumption in~\cite{collevo} which studies the evolution processes for the network of scientific collaboration artifacts. On the other hand,~\cite{slov} provides an empirical study focused on a very specific scientific collaboration network, and ~\cite{distcol} considers a modified preferential attachment algorithm, again for collaboration artifacts along with numerical results. Our work goes beyond these models to provide both a theoretically sound treatment, and accounts for the key phenomenon of subsumption that occurs in various collaborative endeavors. Also, Wu et al.~\cite{Wu2015} have recently proposed a simplicial complex generation model but with a significantly different goal of characterizing the growing geometry of networks. Very recently, growth models for collaboration artifacts have been considered~\cite{clqc}, \eat{after our NetSci and AUTS(abstract online since late'14)} but they address neither the success of matching the size distribution nor subsumption phenomena, which is a focal point of this paper.
\eat{
\emph{While one can consider using the facet size distribution and generate a complex matching that distribution by configuration models, that has no control over the size distribution.}
}

\eat{
Basic structural properties of several collaboration networks, e.g., DBLP and IMDB modeled as simplicial complexes, have been studied recently~\cite{HoangRS14}. It has been observed that there exists a relationship between the number of facets $|F|$ and number of nodes $|V|$ in a simplicial complex $S=(V,F)$ as 
\begin{equation}\label{eq:cbeta}
|F|= c |V|^{\beta}
\end{equation}
Parameters $c, \beta$ can be estimated if data about the longitudinal evolution of the collaboration network is available, as is the case in some publication and movie databases. We exploit this relationship to steer our generative algorithm \GeneSCs to generate simplicial complexes that match real world datasets. \GeneSCs requires as input the facet size distribution, and parameters $c$ and $\beta$. Note that in some cases, longitudinal data is unavailable, so it is assumed that $\beta = 1$.
}

\eat{\textit{While one does not strictly have to use $c$ and $\beta$,...- there are many other alternative ways to using $c$ and $\beta$, however it turns out to be a good fit. -Contant $m$ discussion?}}

\eat{While we do identify the phenomenon of subsumption in our generative algorithm \GeneSCs, we are able to demonstrate that subsumption occurs very infrequently in \GeneSCs, hence we are successful in generating collaboration structures that are close to the real datasets.}

%%%%%%%%%%%%%%%%%%%%%%%%%%%%
\begin{figure*}[htbp]
\begin{algorithmic}[1]
\hrule
\vspace{2pt}
\State \textbf{Algorithm} \GeneSCs$(z(\cdot),c,\beta,|V|)$
\vspace{2pt}
\hrule
\vspace{2pt}
\State $f \leftarrow 0$ 
\Comment{Initialize facet counter}
\While {$|V(f)| \leq |V|$}
\State $s\leftarrow$ {\sc RandomSample}$(z(\cdot))$
\Comment{Randomly generate size of new facet from distribution $z(s)$}
\State $F_f \leftarrow$ {\sc Facet}$(s)$ 
\Comment{Generate $(f+1)$-th facet, with size $s$}
\State $nv = \lceil(\frac{|F(f)|+1}{c})^{\frac{1}{\beta}}\rceil$
\Comment{Target number of nodes after $(f+1)$-th step following growth equation \ref{eq:cbeta}.}
\State $newv = nv - |V(f)|$ \Comment{Number of new nodes introduced, i.e., not to be connected to $S$}
\State $mergev = s - newv$ \label{line:merge}\Comment{Number of nodes in $F_f$ that need to be merged with $S$}
\While {$mergev > 0$} 
\State $u = $ {\sc SelectNode}$(f)$ \Comment{Pick next node in $F_f$}
\State $v = $ {\sc RandomNode}$(S,\frac{f_d(i)}{\sum\limits_{i\in V(f)} f_d(i)})$ \label{line:PA}\Comment{Preferential attachment on facet degree distribution of $S$}
\State Merge nodes $u$ and $v$\Comment{To prevent a given node in existing SC being selected to merge with multiple nodes of the newcomer facet $F_f$, we can sample without replacement, or just re-sample the existing SC until we find a node that has not been picked for merging in this step. Note that the probability of this happening vanishes as $f$ grows. (see Appendix \ref{appx:subsump})}
\State $mergev \leftarrow mergev - 1$ 
\EndWhile
\State $V(f+1) \leftarrow V(f)\:\cup $ {\sc Nodes}$(F_f)$
\State \Comment{The set of facets after checking if $F_f$ subsumes or is subsumed by one or more existing facets in $S$}
\State $F(f+1) \leftarrow$ {\sc FacetSubsumption}$(S,F_f)$ 
\State $S \leftarrow (V(f+1),F(f+1))$
\State $f \leftarrow f + 1$ \Comment{Increment facet counter}
\EndWhile
\State\Return $S=(V(f),F(f))$
\vspace{2pt}
\hrule
\end{algorithmic}
\caption{\GeneSCs algorithm for generating a random Simplicial Complex} \label{alg:GeneSCs}
\end{figure*}

\begin{figure*}[htbp]
\begin{algorithmic}[1]
\hrule
\vspace{2pt}
\State \textbf{Algorithm} {\sc FacetSubsumption}$(S,F_f)$ \Comment{$S$ is represented as a sparse matrix of $|F(f)|$ rows}
\vspace{2pt}
\hrule
\vspace{2pt}
\State $S \leftarrow S \cup F_f$ \Comment{Add a row to $S$ for now}
\For{$j=1$ to $|F(f)|$}
\State $F_j \leftarrow S(j)$\Comment{$j$-th row of $S$}
\If{$|F_f|\geq|F_j|$} \Comment{$F_j$ and $F_f$ are represented as sorted strings of node IDs.}
\If{{\sc Subsequence}$(F_j,F_f)$} \Comment{{\sc Subsequence}$(x,y)$ returns \textbf{true} if $x$ is a sub-string of $y$.}
\State Add $F_j$ to the delete list $D$ \Comment{$F_f$ subsumes $F_j$}
\EndIf
\Else 
\If{{\sc Subsequence}$(F_f,F_j)$}
\State Add $F_f$ to the delete list $D$ \Comment{$F_f$ is subsumed by $F_j$}
\State \textbf{break}
\EndIf
\EndIf
\EndFor
\State $S \leftarrow S \setminus D$ \Comment{Delete all subsumed facets by removing corresponding rows from $S$; update counters}
\vspace{2pt}
\hrule
\end{algorithmic}
\caption{Algorithm for computing subsumptions of facets by / into a Simplicial Complex.} \label{alg:subsump}\vspace{-0.15in}
\end{figure*}

%%%%%%%%%%%%%%%%%%%%%%%%%%%%
\subsection{\GeneSCs: A Facet-based Preferential Attachment Model}
\label{sec:genescs}

\GeneSCs takes as input the facet size distribution $z(s)$ of a real data set (where $s$ is the facet size with $1\leq s\leq |V|$) and generates a random simplicial complex $(V,F)$ modeling that data set with facet degree distribution $p_k$ (where $k$ denotes facet degree). 

It has been observed from recent studies on simplicial complex models of collaboration networks such as DBLP and IMDB~\cite{HoangRS14} that in a simplicial complex $S=(V,F)$, there exists a relationship between the number of facets $|F|$ and number of vertices $|V|$: 
\begin{equation}\label{eq:cbeta}
|F|= c |V|^{\beta}
\end{equation}
Parameters $c, \beta$ can be estimated if data about the longitudinal evolution of the collaboration network is available. We exploit this relationship to steer \GeneSCs to generate simplicial complexes that match real world datasets. If longitudinal data is unavailable, we assume that $\beta = 1$; then $c$ is the \emph{facet density}.

The algorithm (shown in pseudo-code form in Figure \ref{alg:GeneSCs}) grows the simplicial complex $S$ by adding one randomly generated facet at a time, decides points of attachment in $S$, and checks for subsumption events. Note that before the addition of a facet indexed by integer $f$ and denoted by $F_f$, the state of the simplicial complex is denoted by $S(f)=(V(f),F(f))$.

\GeneSCs is computationally efficient. We store the Simplicial Complex in a sparse matrix with $|F(f)|$ rows and the number of columns is a maximum facet size encountered so far. The biggest computational bottleneck is the PA step in line \ref{line:PA}. To speed up the computation, we use the following identity relating the sum of facet degrees to the sum of facet sizes, a quantity which is easy to update after every step. Here, $f_d(u)$ is the facet degree of node $u$ and $|f|$ is the size (number of vertices) in facet $f$.
\baq
\sum\limits_{u\in V(f)} f_d(u) = \sum\limits_{f\in F(f)} |f|
\eaq
Another source of speedup is computing subsumptions by solving a problem of matching two small substrings corresponding to facets (small relative to $|V|$) as shown in Figure \ref{alg:subsump}.

\GeneSCs has several distinctions from other works that have proposed PA based generative growth models for hypergraphs~\cite{HAMND2011,Liu2012}. Since we are interested in modeling the global collaboration relation and not its artifacts, in our model the \emph{hyperedges} are subsumed to yield \emph{facets}, thus preserving the core structure underlying the participation of various individuals in a collaboration. We do not assume that the size of such collaborative structures is power law distributed. In fact, this is not the case for several collaboration networks, especially near the head, i.e., small collaborations. Accordingly, \GeneSCs takes \emph{as input} the distribution of facet sizes and average facet density (i.e., average number of collaborations per node) to generate a collaboration relation using a variant of PA.

The dynamics of \GeneSCs is distinct from classical PA in two ways. First, the structural unit of growth in each time step in our setting is a \emph{facet} which could contribute one or more nodes to the simplicial complex. In contrast, in classic PA, one node is added in each time step. Secondly, in classical PA new nodes are always added to the network, whereas in our case, some nodes in the newly generated facet may be merged with existing nodes; moreover, new facets may \emph{subsume} older facets or \emph{be subsumed} by older facets. This is consistent with how large collaboration networks grow -- new endeavors consist of both existing individuals and new individuals.

It is well known that preferential-attachment (PA) based network growth methods result in power law vertex degree distributions for classical graphs~\cite{DMS2000}. Other more general variants of preferential attachment have been analyzed extensively as well~\cite{Krapivsky2001}. We show below that the growth model behind \GeneSCs results in SCs with power law distributed facet degree. We also compute the power law exponent as a function of input parameters such as average facet size and average facet density, and show that it matches real world collaboration network data sets well.

\begin{theorem}[Facet degree properties]
\label{thm:fdeg}
If the average facet density (the average number of facets each node belongs to) is $c$ and average facet size is $s$, \GeneSCs generates a random simplicial complex whose facet degree distribution is power law with exponent $\alpha=2 + \frac{1}{c\:s-1}$.
\end{theorem}
\begin{proof}
We use mean field arguments in this proof. Let the current number of nodes and facets in the simplicial complex (SC) be denoted by $n$ and $f$, respectively. Since $c$ is the facet density, we have $c=\frac{f}{n}$, at least when the SC has grown large in size. At the current time step, a new facet arrives into SC and the facet count becomes $f+1$. Simultaneously, the node count is expected to increase to $\frac{f+1}{c}=n+\frac{1}{c}$ (Note that for the purpose of clarity, throughout the analysis we ignore the effects of rounding to integer values for some of the variables.)

Let $p_k(f)$ be the fraction of nodes in SC with facet degree (fdegree) $k$ when there are $f$ facets in the SC. If node $i$ has fdegree $k_i$, then the PA step in \GeneSCs will merge a node from the newly arriving facet into node $i$ with probability $p_i=\frac{k_i+a}{\sum\limits_{i=1}^n (k_i+a)}=\frac{k_i+a}{\sum\limits_{i=1}^n k_i+ a\:n}$, where $a$ is the initial attractiveness parameter \cite{DMS2000}. It can be observed that if facet sizes are given by $s_j, j\in [1,f]$, we have $\sum\limits_{i=1}^n k_i = \sum\limits_{j=1}^f s_j = f\:s$. Therefore, $p_i = \frac{k_i+a}{f\:s + a\:n} = \frac{k_i+a}{c\:n\:s + a\:n}=\frac{k_i+a}{(c\:s+ a)\:n}$.

When the $(f+1)$-th facet of average size $s$ is added to SC, on average the number of new nodes that are added to SC are $\frac{f+1}{c}-\frac{f}{c}=\frac{1}{c}$. Therefore, \GeneSCs attempts to merge $s-\frac{1}{c}$ nodes (on average) in the new facet with old nodes in SC by performing PA independently for each such node. 

Just like in the regular PA for graphs, the probability of more than one node getting merged with a single old node in SC goes is vanishingly small as $n, f \to \infty$, hence we assume that each of the $s-\frac{1}{c}$ nodes in the new facet get merged to distinct nodes in SC with high probability (Please see Lemma \ref{appA1}). Since there are $n p_k(f)$ nodes in the SC with fdegree $k$, the expected number of new collaborations (this new facet) picked up by all nodes of fdegree $k$ in SC as a result of the addition of the new facet is given by $n p_k(f) \frac{k+a}{(c\:s+ a)\:n} (s-\frac{1}{c}) = p_k(f) \frac{k+a}{c\:s+ a} (s-\frac{1}{c})$. Thus, the expected number of nodes in SC whose fdegree becomes $k+1$ as a result of the facet arrival is given by $p_k(f) \frac{k+a}{c\:s+ a} (s-\frac{1}{c})$.

We observe that for each node with fdegree $k-1$ that gets merged with the new facet, the number of collaborations increases by one, thus increasing their fdegree to $k$. Applying the above reasoning, the expected number of such collaborations is thus $p_{k-1}(f) \frac{k-1+a}{c\:s+ a} (s-\frac{1}{c})$. Also, the expected number of nodes with fdegree $k$ after the addition of the new facet is $(n+\frac{1}{c})p_k(f+1) = \frac{f+1}{c} p_k(f+1)$, since there are $n+\frac{1}{c}$ nodes in the SC at this stage.

It can be shown that the facet count increases by one at least for large SCs, since under \GeneSCs the probability of a newly arriving facet subsuming an existing facet or getting subsumed becomes vanishingly small as $n, f \to \infty$. See Lemmas \ref{lem:subs-prob}, \ref{lem:subs-lessmult}, \ref{lem:subs-v-size}, \ref{lem:subs-gt1}, \ref{lem:App6}, and Remark \ref{rem:App7} in Appendix \ref{appx:subsump} for details.
The fact that the amount of subsumption resulting from \GeneSCs is small is desirable since we utilize the facet size distribution $f(s)$ as an input parameter, and all the hyperedge to facet subsumption is already captured in $f(s)$. Therefore, one can set up the ``master equation" that results from the conservation of collaboration counts after the addition of the $(f+1)$-th facet into SC:
\begin{align}\label{eq:meq}
\nonumber(n\!+\!\frac{1}{c})p_k(f\!+\!1) &= n p_k(f) \!+ \frac{k-1+a}{c\:s+ a} (s\!-\!\frac{1}{c}) p_{k-1}(f) \\
&- \frac{k+a}{c\:s+ a}\:(s-\frac{1}{c})\:p_k(f)
\end{align}

Substituting $n = \frac{f}{c}$ into Eq. (\ref{eq:meq}), we get the master equation as a function of $f$ alone.
\begin{align}\label{eq:meqk}
\nonumber(f\!\!+\!\!1)p_k(f\!\!+\!\!1) &= f p_k(f) + \frac{(k\!-\!1\!+\!a)(c\:s-1)}{c\:s+ a} p_{k-1}(f)\\
&- \frac{(k+a)(c\:s-1)}{c\:s+ a}\:p_k(f)
\end{align}

Note that for $k=1$ (the lowest fdegree in SC), Equation (\ref{eq:meqk}) is not accurate since there is no dependence on $p_{k-1}(f)$. Instead, the new facet has a contribution of $\frac{1}{c}$ new nodes of fdegree 1, on average. We reflect this in the following equation for $k=1$:
\begin{align}\label{eq:meq1}
\nonumber(f+1)p_1(f+1) &= f p_1(f) + 1 \\
&- \frac{(1+a)(c\:s-1)}{c\:s+ a}\:p_1(f)
\end{align}

Assuming that $p_k(f)$ converges to $p_k$ when $f\to\infty$ for all $k$, we rewrite Equation (\ref{eq:meq1}) and substitute $g = c\:s$ to get:
\begin{align}
\nonumber p_1 &= 1 - \frac{(1+a)(g-1)}{g+a}\:p_1 \\
\Rightarrow p_1 &= \frac{g+a}{g(1+a) + g-1}
\end{align}

Performing similar transformations to Equation (\ref{eq:meqk}), we get
\begin{align}
\nonumber p_k &= \frac{k-1+a}{k+\frac{g(1+a)}{g-1}}\: p_{k-1}\\
&= \frac{(k-1+a) \cdots (1+a)}{(k+\frac{g(1+a)}{g-1}) \cdots (2+\frac{g(1+a)}{g-1})} \cdot \frac{g+a}{g-1+g(1+a)}
\end{align}

Using the basic recurrence for Gamma functions $\Gamma(x+1)=x\Gamma(x)$, we have the identity $\frac{\Gamma(x+n)}{\Gamma(x)} = (x+n-1)(x+n-2)\cdots x$. Using this identity, we have:
\begin{align}\label{eq:pk}
\nonumber p_k &= \frac{\Gamma(k+a)\:\Gamma(2+\frac{g(1+a)}{g-1})}{\Gamma(1+a)\:\Gamma(k+1+\frac{g(1+a)}{g-1})}\cdot\frac{g+a}{1+\frac{g(1+a)}{g-1}}\\
\nonumber &= \frac{\Gamma(k+a)\:\Gamma(1+\frac{g+a}{g-1})}{\Gamma(k+1+\frac{g(1+a)}{g-1})} \cdot \frac{\Gamma(2+\frac{g(1+a)}{g-1})}{\Gamma(1+a)\:\Gamma(1+\frac{g+a}{g-1})} \\
\nonumber &\times \frac{g+a}{g-1+g(1+a)}\\
&= \frac{B(k+a,1+\frac{g+a}{g-1})}{B(1+a,1+\frac{g+a}{g-1})} \cdot \frac{g+a}{g-1+g(1+a)},
\end{align}
where $B(x,y)=\frac{\Gamma(x)\Gamma(y)}{\Gamma(x+y)}$ is the Beta function.

For large $x$, $B(x,y)$ exhibits power law behavior; specifically, $B(x,y) \approx x^{-y} \Gamma(y)$. In Equation (\ref{eq:pk}), the only term that is dependent on $k$ is $B(k+a,1+\frac{g+a}{g-1})$. Applying the aforementioned power law approximation for large $k$, we get:
\begin{align}
\nonumber p_k &\sim (k+a)^{-(1+\frac{g+a}{g-1})} \sim k^{-(1+\frac{g+a}{g-1})} \\
&= k^{-(2+\frac{1+a}{g-1})} = k^{-(2+\frac{1+a}{c\:s-1})}
\end{align}
Therefore, the facet degree of a large SC generated by \GeneSCs is power law distributed with exponent $2+\frac{1+a}{c\:s-1}$. Since we set the attractiveness parameter $a=0$ in the default mode of \GeneSCs, the result follows.
\end{proof}

Note that the denominator of the exponent $c\:s-1$ is strictly positive for simplicial complexes of interest. This is because $c\:s=\frac{s\:f}{n}=\frac{\sum\limits_{j=1}^f s_j}{n}=\frac{\sum\limits_{i=1}^n k_i}{n}$. Since $\forall i: k_i \geq 1$, for any non-degenerate SC that is not a disjoint union of full dimensional simplexes, $\sum\limits_{i=1}^n k_i > n$, and therefore $c\:s -1 > 0$.

The quantity $c\:s$ can be interpreted as the average number of \emph{collaborators} of a node, while counting each collaborator distinctly for every new collaboration. This is distinct from the average number of collaborators of a node\textit{, which is given by the average node degree}.

%%%%%%%%%%%%%%%%%%%%%
\begin{figure*}[!t]
\centering
\includegraphics[trim=0 2.5in 0 2.5in,width=0.5\textwidth]{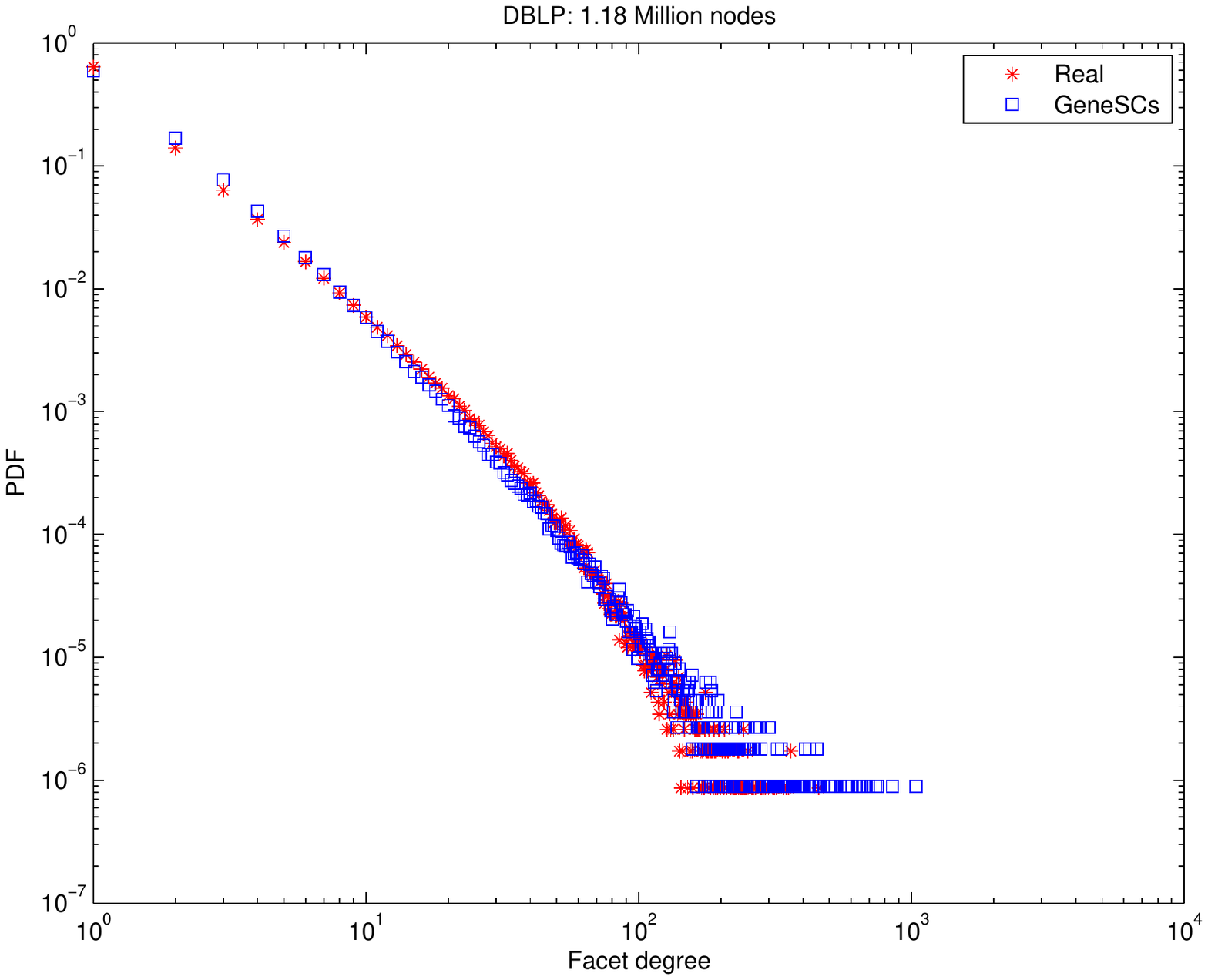}
\includegraphics[width=0.48\textwidth]{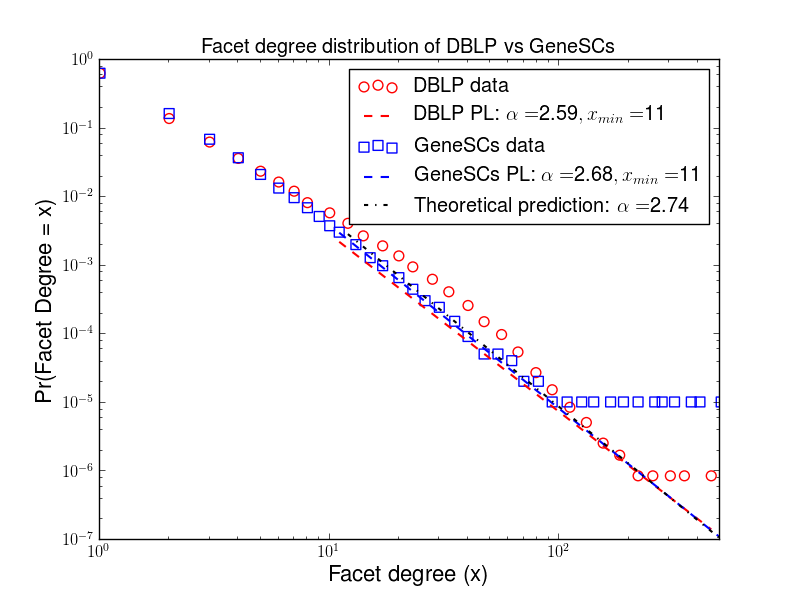}
\caption{\label{fig:DBLP-performance}Performance of \GeneSCs vs. real DBLP collaboration data. (a) $D_{TV}\approx 0.035$. (b) Tail behaviors of various power law curve superimposed on logarithmic binned data. Since $c=0.68, s=3.44$, the analytical prediction of the power law exponent of \GeneSCs is $\alpha \approx 2.74$. The noisy tails are due to a small number of samples, which is typical of logarithmic binning.}
\end{figure*}
%%%%%%%%%%%%%%%%%%%%%
\begin{figure*}
\centering
\includegraphics[scale=0.29]{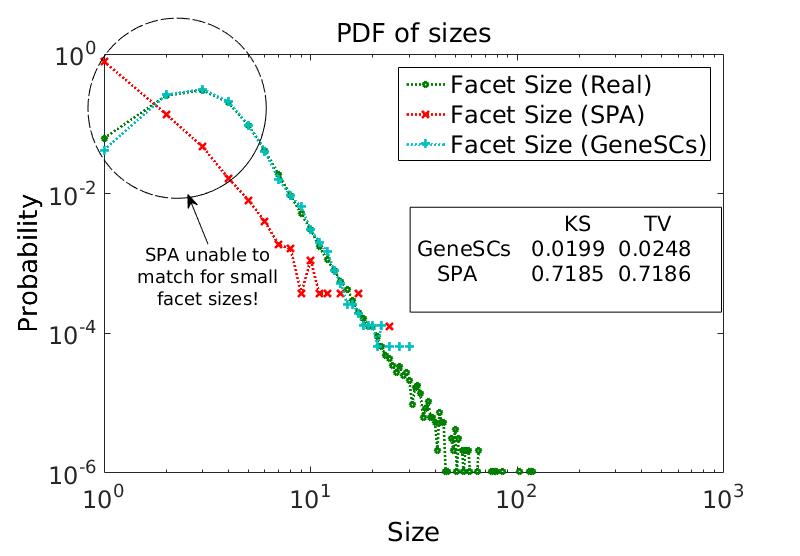}
\includegraphics[scale=0.29]{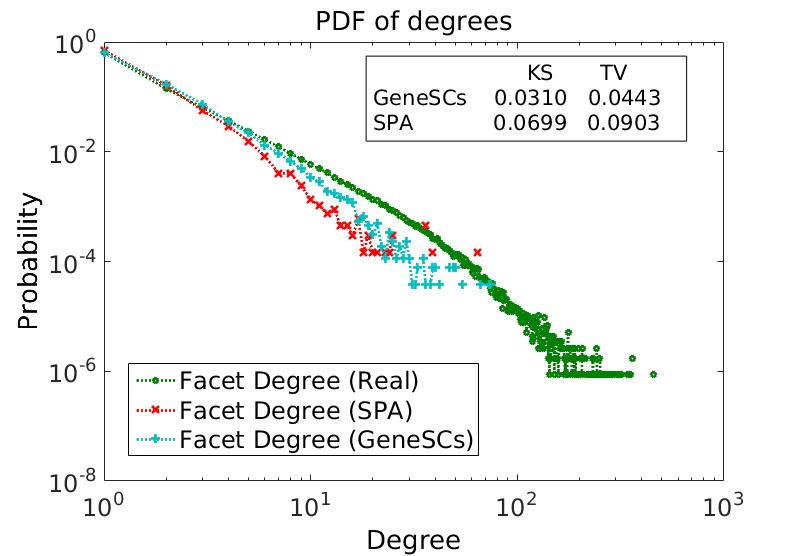}
\caption{\label{fig:gsc-v-spa} Comparison with SPA algorithm for DBLP}
\end{figure*}
%%%%%%%%%%%%%%%%%%%%%

\subsection{Performance evaluation of \GeneSCs}
\label{sec:eval}

We measured the quality of the distribution generated by \GeneSCs using the Kolmogorov-Smirnov distance $D_{KS}$ and the Total Variation distance $D_{TV}$ between real data distributions $p(\cdot)$ and the generated distributions $q(\cdot)$ for both facet sizes and facet degrees: 
\begin{eqnarray}
D_{KS}(p||q) & = & \sup_x |p(x)-q(x)|\\
D_{TV}(p||q) & = & \frac{1}{2} \sum_x |p(x)-q(x)|
\end{eqnarray}

Note that the facet sizes in \GeneSCs are generated using the facet size distribution of real collaboration data sets, where the effect of subsumptions has already been incorporated in the first place. We observed that \GeneSCs yields $D_{KS}\approx 0$ and $D_{TV}\approx 0$ for facet sizes, confirming the fact that subsumption events are rare if one is drawing random collaboration structures from the \emph{facet} size distribution instead of the \emph{hyperedge} size distribution. Analytic arguments for this phenomenon are presented in Appendix \ref{appx:subsump}.

Figure \ref{fig:DBLP-performance} illustrates the performance comparison of \GeneSCs with respect to real DBLP publication data and the theoretical prediction of Theorem \ref{thm:fdeg} as far as the facet degree distribution is concerned. It can be observed that \GeneSCs matches well the characteristics of both the real facet degree distribution and the theoretical prediction.

Figure \ref{fig:gsc-v-spa} compares the relative performance of the Structured Preferential Attachment (SPA) algorithm~\cite{HAMND2011,HAMND2012} with \GeneSCs. When applying SPA, first the best parameters that SPA fit for DBLP's hyperedge degree and size distributions are obtained and used for hyperedge generation. Then, we have inflicted subsumptions on that data and plotted it. It can be observed from Figure \ref{fig:gsc-v-spa}(a) that SPA, which is designed to generate pure power law distribution for both hyperedge sizes and hyper-degrees, is unable to match the real facet size distribution of the DBLP data set $(D_{TV}=0.719)$. It is also unable to closely match the facet degree distribution of DBLP $(D_{TV}=0.09)$, especially near the heavy tail. Moreover, it has a markedly different slope. In contrast, since \GeneSCs samples the facet size distribution, it is able to yield a close match  to (obviously) that distribution $(D_{TV}=0.0248)$ (the minor difference is due to the occurrence of some subsumptions at low sizes), and the facet degree distribution $(D_{TV}=0.044)$.

\eat{
\subsection{Temporal growth properties}
\textcolor{red}{(Need to put something about the temporal growth models)}
}
%%%%%%%%%%%%%%%%%%%%%%%%%%%%%
\subsection{Variants of \GeneSCs: Smoothed Preferential Attachment}
\label{sec:smoothPA}

While \GeneSCs performs notably well for DBLP using pure preferential attachment (PA) on facet degrees, it is unable to generate the facet degree distributions for collaboration networks such as Physical Review D (PRD) and IMDB. Hence, we propose two variants of PA to address this drawback.

\vspace{1ex}
\noindent\textbf{Clamped Preferential Attachment}
While PA is likely to be a basic force behind collaboration network formation, some domains leverage other peculiar behaviors while forming large collaborative structures. For example, the circumstances and motives behind the formation of ``organic" academic collaborations are typically different from those driving the formation of collaborations in entertainment or artistic fields (e.g. movies). Even in academia, the way research is conducted highly depends on the particular sub-field, as theoretic and experimental communities have different processes of forming groups and performing collaborative research. Hence, it is not surprising that \GeneSCs is not a one-size-fits-all solution. For instance, the maximum facet degree of the generated SC tends to exceed the maximum facet degree of the real data set. Moreover, the facet degree distributions of the generated and real SC have an unacceptably high total-variation distance, $D_{TV} = 0.2$.

Reflecting on this behavior, we posit the following hypothesis for real networks: individuals typically do not categorize the popularity of other individuals precisely by the number of connections; rather, they might make better assessment regarding people with \emph{similar} popularity, i.e., they have a coarser perception, particularly for very popular individuals. This is indeed an issue with collaborations in experimental physics, which tend to publish in Physical Review D. There are several papers reporting results from large experiments with very large sets of co-authors (See Figure \ref{fig:PRD-HGvsSC}). In such networks, a coarse grained view of popularity is likely to better explain network growth.

To model this hypothesis, we modify \GeneSCs in the PA step (particularly defined on line \ref{line:PA} of Algorithm  \GeneSCs). More precisely, rather than using the exact facet degree, we \emph{clamp} the facet degrees of each node in order to smooth the perception of popularity of the very high degree nodes. In the most basic form, one can achieve this as follows:
\begin{equation}
f_{cld}(u) = \min(f_d(u), F_d^{clamp}),
\end{equation}
where $f_d(u)$ is the actual facet degree of node $u$, $F_d^{clamp}$ is the maximum clamp value and $f_{cld}(u)$ is the clamped facet degree which is input to the PA step of \GeneSCs.

%%%%%%%%%%%%%%%%%%
\begin{figure}[ht]
\centering
\vspace{-0.17in}\includegraphics[width=1.0\columnwidth]{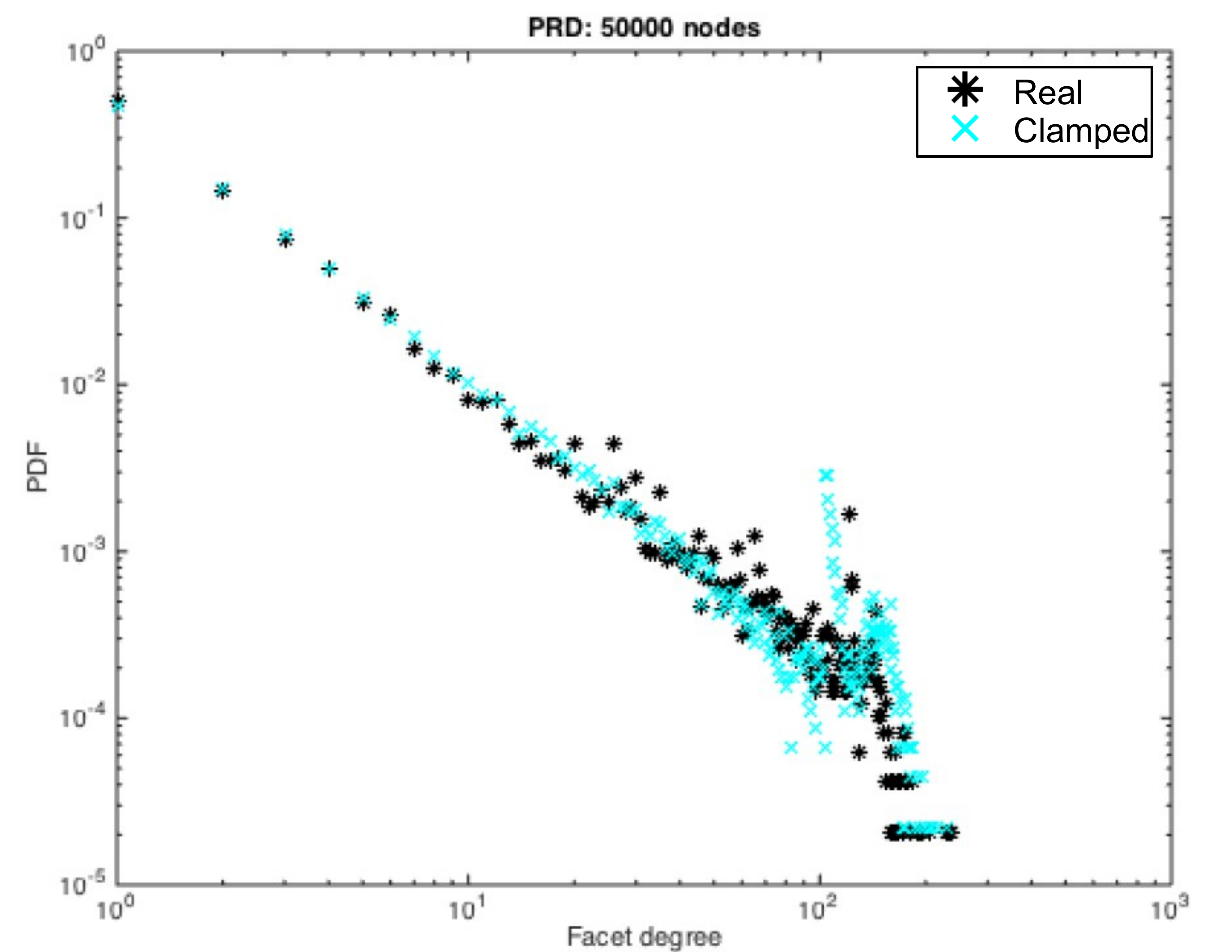}
\caption{\label{fig:clampPA} Clamped PA: Phys Rev D data ($\alpha=0.25$): $D_{TV} = 0.035$}\vspace{-0.1in}
\end{figure}
%%%%%%%%%%%%%%%%%%

While this method classifies many very popular individuals as \emph{popular}, it is likely to be limited, since too much granularity is lost. Consequently, we actually use a softer mapping (unlike the step function mentioned earlier) as follows:
\begin{equation}\label{pasclam}
f_{cld}(u)=\begin{cases}f_d(u), &\hbox{if}\: f_d(u)<\alpha f_d^{\max}\\
\sqrt{\alpha f_d(u) f_d^{\max}}, &\hbox{if}\: f_d(u)\geq \alpha f_d^{\max}
\end{cases},
\end{equation}
where $\alpha$ is a design parameter, and $f_d^{\max}$ is the maximum of actual facet degrees at the current step, i.e., $f_d^{\max}=\max_{u} f_d(u)$.

Effectively, this assignment \emph{still} differentiates among the \emph{very popular} nodes, but smooths their relative popularity values considerably.

After this mapping is performed, \GeneSCs operates using these alternative facet degrees. For the sake of completeness, line \ref{line:PA} of Algorithm \GeneSCs is replaced by the following modified-PA rule:
\begin{equation}\label{pagen}
p(u)= \frac{f_{cld}(u)}{\sum\limits_{u \in V} f_{cld}(u)},
\end{equation}
with $f_{cld}(u)$ obtained through (\ref{pasclam}). Here we note that while modifying standard PA has been considered in \cite{npacol}, which propose a nonlinear preferential attachment by replacing node degree $f_d(u)$ by $f_d(u)^{\beta}$ for each node in the connection process with some given $\beta$ , our clamped PA defined by the mapping in (\ref{pasclam}) is significantly different and also depends on many distinct parameters as $\alpha$ and $f_d^{\max}$.   

Figure \ref{fig:clampPA} shows that the clamped PA can yield a close match to the facet degree distribution in the PRD data set, including the spikes at the tail.

%%%%%%%%%%%%%%%%%%
\begin{figure}[ht]
\centering
\includegraphics[width=1.0\columnwidth]{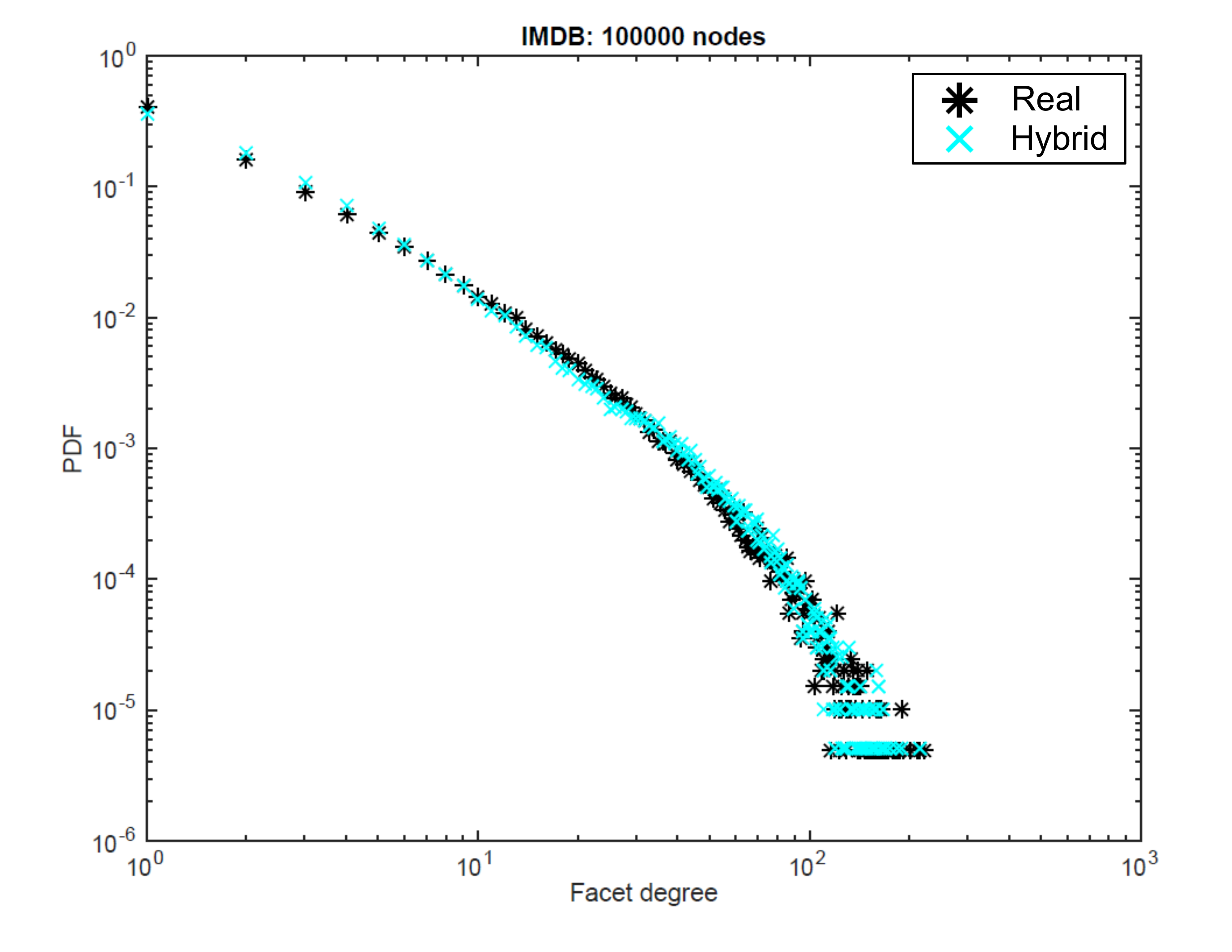}
\vspace{-0.3in} \caption{\label{fig:hybrid} Hybrid PA and uniform attachment: IMDB data $(T=12,\alpha=0.25)$: $D_{TV} = 0.04$}
\end{figure}
%%%%%%%%%%%%%%%%%%

\vspace{1ex}
\noindent\textbf{Hybrid Uniform and Preferential Attachment}

Another behavior that we intuitively expect in building collaboration structures is that it is rarely the case that every member of a large collaboration is well-known. In other words, typically a limited number of collaborators are popular and the remaining ones are not really well-known or popular. For instance, consider a movie cast. It is often the case that only a small subset (e.g. 10-15) of the whole movie cast constitutes well-known actors, which have more screen time and central roles (and hence are immediately identifiable by an average movie goer), while the rest are mostly figures with minor side roles.

Based on the above intuition, particularly for IMDB style collaborations, we propose the following variant to \GeneSCs: At line \ref{line:merge} of the algorithm, out of $mergev$ connections to the existing simplicial complex, only the first $T$ (or $min(mergev,T)$) nodes to merge are determined by applying PA on the facet degrees. The remaining $mergev-T$ nodes (assuming $mergev>T$) to be merged are selected \emph{uniformly} at random, regardless of their facet degrees. We observe that this \emph{preferential-uniform hybrid} \GeneSCs along with clamped facet degrees for the PA sub-routine provides a much better match to the distributions of actual IMDB dataset, thus verifying our qualitative intuitions. This can be observed in Figure \ref{fig:hybrid}.

While the idea of combining preferential attachment and uniform attachment has been considered before~\cite{unipahyb}, it has been done so only for graphs. Moreover, the latter approach \emph{probabilistically mixes} the two methods for every connection, in contrast to our threshold-based approach for generating SCs.

\eat{ FD: 0.0232, 0.1136; FS: 0.0344, 0.1257 for monotonic sqrt variant with $\alpha=0.25$ and PRD.}
\section{Concluding Remarks}

In this paper, we proposed \GeneSCs, a generative model for collaboration structures modeled as simplicial complexes (SC). SCs are different from graphs since they have two different dimensions for growth (facet size and facet degree), whereas graphs only have one (node degree), since all edges have identical sizes. SCs are also different from hypergraphs, which do not have the subsumption property. While hypergraphs are good for modeling artifacts of collaborations, e.g., papers and movies, SCs are more appropriate for succinctly modeling the inherent structure of the collaboration relationship, i.e., who all have collaborated with each other. This distinction is key because subsumptions are very common in real-world collaboration networks.

The facet size distribution constrains the facet degree distribution to an extent. Leveraging this observation, \GeneSCs takes facet size distributions of real world collaboration networks as input and efficiently generates random SCs with facet degree distributions matching those corresponding to the real data. Our theoretical analysis is shown to accurately predict the statistical properties of SCs generated by \GeneSCs in its pure form. For collaborations that have characteristics such as very heavy tails of facet sizes or non-power law popularity distributions of participants as collaborators, appropriate modifications to the preferential attachment step of \GeneSCs (namely, clamping and uniform-hybridization) yields good intuitively sound results.
\appendix
\section{Characterizing the Probability of Facet Subsumption}
\label{appx:subsump}

In this appendix, we consider various situations in which a newcomer facet subsumes one or more facets in the existing SC when following the rules of \GeneSCs. In the following lemmas, we show that the probabilities of such subsumptions become vanishingly small as the SC grows in size. Similar reasoning can be applied for the reverse case, where a newcomer facet is subsumed by an existing facet in the SC.

\begin{lemma}\label{appA1}
The probability of more than one node getting merged with a single existing node in the SC generated thus far becomes vanishingly small as $n, f \to \infty$, where $n$ and $f$ denote the number of nodes and facets in the SC, respectively.
\end{lemma}
\eat{
The main mechanism of demonstrating power law behavior is expressing balance equations of number of nodes with specific degree~\cite{DMS2000}. At any generation step $t>1$, the degree of a node of degree $k$ mainly varies as follows: (i) Either it stays same, as no edges are connected to this node -- the probability of this event decreases with $k$; (ii) It increases to $k+1$, with probability proportional to $k$; (iii) There is also the probability that it exceeds $k+1$, which means multiple edges are connected to this node. However, it can be shown that the probability of the last event is $O(\frac{k^2}{t^2})$, i.e., vanishingly small, as the generation algorithm proceeds, and has no effect on the recursive equations leading to power law node degree.

Let $p(k,s,t)$ be the probability that node $s$ has facet degree $k$ at time instant $t$.  If at each step, a newly generated facet is attempting to make $m$ connections with the existing simplicial complex, we can express $p(k,s,t)$ recursively in a master equation form as in ~\cite{DMS2000}.

\baq
p(k,s,t+1) & = & \left(1-\frac{k}{2mt}\right)^m p(k,s,t) + \binom{m}{1} \frac{k-1}{2mt} \left(1-\frac{k-1}{2mt}\right)^{m-1} p(k-1,s,t) \\
	& + & \binom{m}{2} \left(\frac{k-2}{2mt}\right)^2 \left(1-\frac{k-2}{2mt}\right)^{m-2} + \cdots \\
	& \approx & \left(1-\frac{k}{2t}\right) p(k,s,t) + \frac{k-1}{2t} \left(1-\frac{(k-1)(m-1)}{2mt}\right) p(k-1,s,t) + O(\frac{p}{t^2})\label{xapp}\\
	& = & \left(1-\frac{k}{2t}\right) p(k,s,t) + \frac{k-1}{2t}  p(k-1,s,t) + O(\frac{p}{t^2})
\eaq
}

\begin{proof}
Assume that at the current facet-addition step $f$ of \GeneSCs (at this step, SC is supposed to have $f$ facets), the newly generated facet of (average) size $s$ is to be merged with the existing SC at $m\approx s-\frac{1}{c}$ nodes. Let $Z_k$ be a random variable denoting the number of merges of an existing node (say, $X$) of facet degree $k$ with one or more distinct nodes $Y_1, Y_2, \ldots$ of the incoming facet. For large enough $f$, the probability that $Y_i$ gets merged into $X$ is $\frac{k}{f s}$ following the Preferential Attachment rule ($X$ has facet degree $k$ and there are $f s$ nodes in the current SC). Assuming $m$ independent Bernoulli trials for merges (this is reasonable in the large network limit as in ~\cite{DMS2000}), the probability of merging $l$ times with $X$ can be expressed as a Binomial random variable. Accordingly, the aggregate probability of having $l$ merges is $\binom{m}{l}(\frac{k}{f s})^l (1-\frac{k}{f s})^{m-l}$.

\begin{align}
P(Z_k > 1)=\sum_{l=2}^m \binom{m}{l}\bigg(\frac{k}{f s}\bigg)^l \bigg(1-\frac{k}{f s}\bigg)^{m-l}
\end{align}
\begin{align}
P(Z_k = 2)&= \frac{m(m-1)}{2} \bigg(\frac{k}{f s}\bigg)^2 \bigg(1-\frac{k}{f s}\bigg)^{m-2}\\
\label{xapp}&<\frac{s(s-1)}{2} \bigg(\frac{k}{f s}\bigg)^2 \bigg(1-\frac{k}{f s}\bigg)^{m-2}\\
&\leq \frac{k^2}{2 f^2} = O\bigg(\frac{k^2}{f^2}\bigg),
\end{align}
where in (\ref{xapp}) we have used the well-known Bernoulli inequality $(1-x)^a \leq 1-xa$ for $x<1$, since $\frac{k}{f s} \rightarrow 0$ as $f \rightarrow \infty$. It can be readily shown that $P(Z_k = l)=O(\frac{k^l}{f^l})$, hence the probability of multiple edges merging to a given node with degree $k$ is $O(\frac{k^2}{f^2})$.
\end{proof}

\vspace{1ex}
%%%
We propose the \emph{supernode method} to establish upper bounds on the subsumption probability. In order to investigate whether a given facet is subsumed by the new-coming facet at step $f$, we form a \emph{virtual supernode} which models the nodes of the facet jointly as shown in Figures \ref{fig:Rapp1}-\ref{fig:Rapp2}. More specifically, the facet degree of the supernode is assigned to be the sum of the facet degrees of the individual nodes.

%%%%%%%%%%%%%%%%%
\begin{figure}[ht]
\centering
\includegraphics[width=0.6\columnwidth]{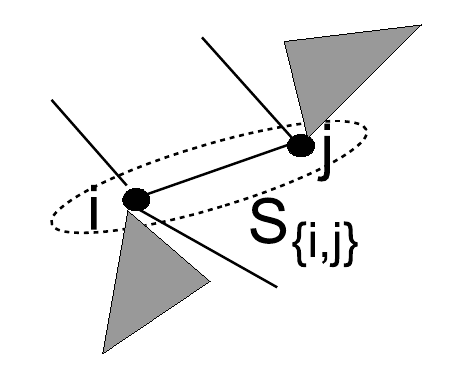}
\caption{Supernode of facet of dimension 2.}\label{fig:Rapp1}
\end{figure}\vspace{-0.1in}
%%%%%%%%%%%%%%%%%
\begin{figure}[ht]
\centering
\includegraphics[width=0.6\columnwidth]{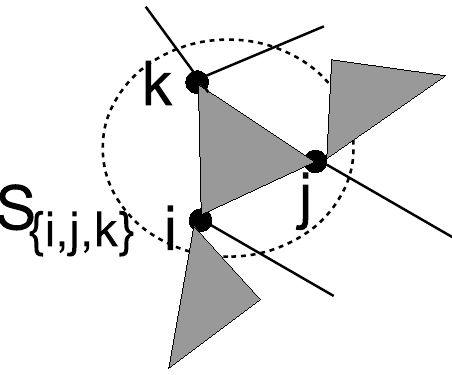}\vspace{-0.1in}
\caption{Supernode of facet of dimension 3.}\label{fig:Rapp2}
\end{figure}\vspace{-0.2in}
%%%%%%%%%%%%%%%%%

\begin{lemma}\label{lem:subs-prob}
For any scenario where the facet dimension is bounded, the probability of subsumption is $O(\frac{1}{f^2})$, and thus negligible in large SCs.
\end{lemma}
\begin{proof}
Consider the probability of a facet being subsumed. For instance, when the facet under consideration is an edge ($i,j$), this means two distinct nodes $s$ and $t$ of the $m$ nodes of the newcomer facet $F_f$ to be merged into the existing network are specifically merged to $i$ and $j$. This can happen as ($s \to i$ and $t \to j$) or ($s \to j$ and $t \to i$).

Now define the supernode $K$ comprised of edge $(i,j)$. Let us consider the probability that $K$ is selected for merging to more than one node of $F_f$ (Fig. \ref{fig:Rapp3}). To analyze this, $K$ can be treated as an ordinary node with degree $d_i+d_j$. \eat{It can be shown that} From Lemma \ref{appA1}, with preferential attachment, the probability of any node $v$ with degree $d_v$ getting multiple merges is $O(\frac{d_v^2}{f^2})$. We readily utilize this result to characterize the probability of $K$ receiving multiple merges as $O(\frac{(d_i+d_j)^2}{f^2})$.
\end{proof}

\begin{figure}[ht]\vspace{-0.1in}
\centering \centering
\includegraphics[width=0.6\columnwidth]{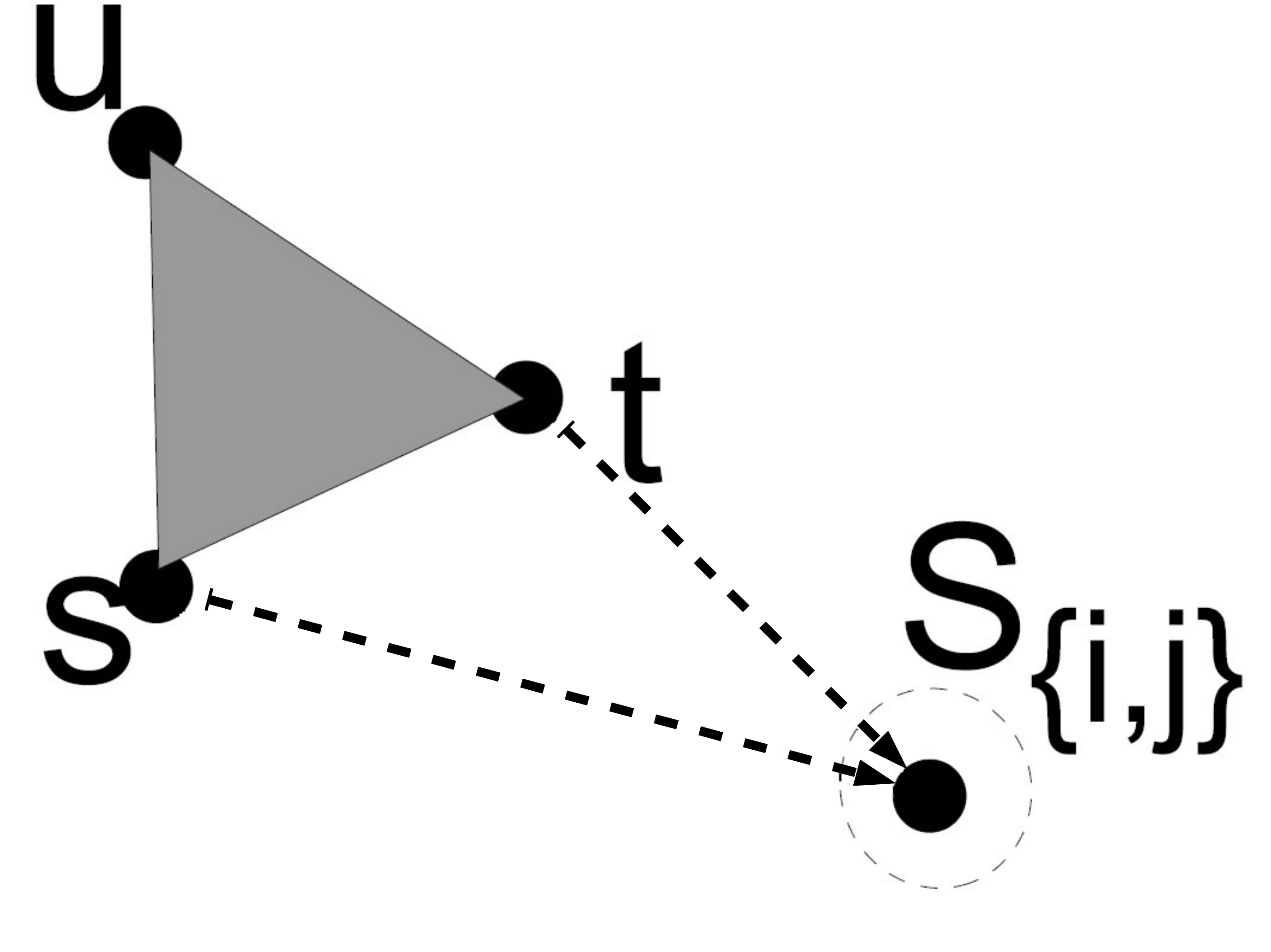}
\caption{A supernode getting multiple merges. Here the newcomer facet $(s,t,u)$ subsumes edge $(i,j)$ in the existing SC.}\label{fig:Rapp3}
\end{figure}\vspace{-0.2in}

\begin{lemma}\label{lem:subs-lessmult}
The probability of facet subsumption is less than the probability of the corresponding supernode getting multiple merges.
\end{lemma}
\begin{proof} As an example, let us consider the case when the facet under consideration is an edge ($i,j$), which means two distinct nodes $s$ and $t$ of the $m$ nodes of the newcomer facet $F_f$ to be merged into the existing network are specifically merged to $i$ and $j$. This can happen as ($s \to i$ and $t \to j$) or ($s \to j$ and $t \to i$). Now define the supernode $K$ comprised of edge $(i,j)$. Supernode $K$ getting multiple merges can occur in four combinations as shown in Figures \ref{fApp4} and \ref{fApp5}: ($s \to i$ and $t \to j$) or ($s \to j$ and $t \to i$), and both nodes selecting to merge to the same nodes ($s \to i$ and $t \to i$) or ($s \to j$ and $t \to j$). Hence, the set of events corresponding to supernodes getting multiple merges is a superset of the set of events corresponding to actual subsumption of the corresponding facets.

Next, consider (an existing) facet $F^{(L)}$ of dimension $L$, with $m\geq L$ nodes of the newcomer facet $F_f$ being merged to the existing SC. Define the supernode corresponding to $F^{(L)}$ as $K_L$. Now, $L$ distinct nodes of $F_f$ can  merge to the supernode $K_L$ in $L^L$ distinct combinations which lead to a multiple edge merge to $K_L$. On the other hand, only $L!$ of these combinations, i.e., a permutation of the $L$ distinct nodes of $F^{(L)}$ result in an actual subsumption.
\end{proof}

%%%%%%%%%%%%%%%%%% 
\begin{figure}[ht]
\centering
\includegraphics[width=1.0\columnwidth]{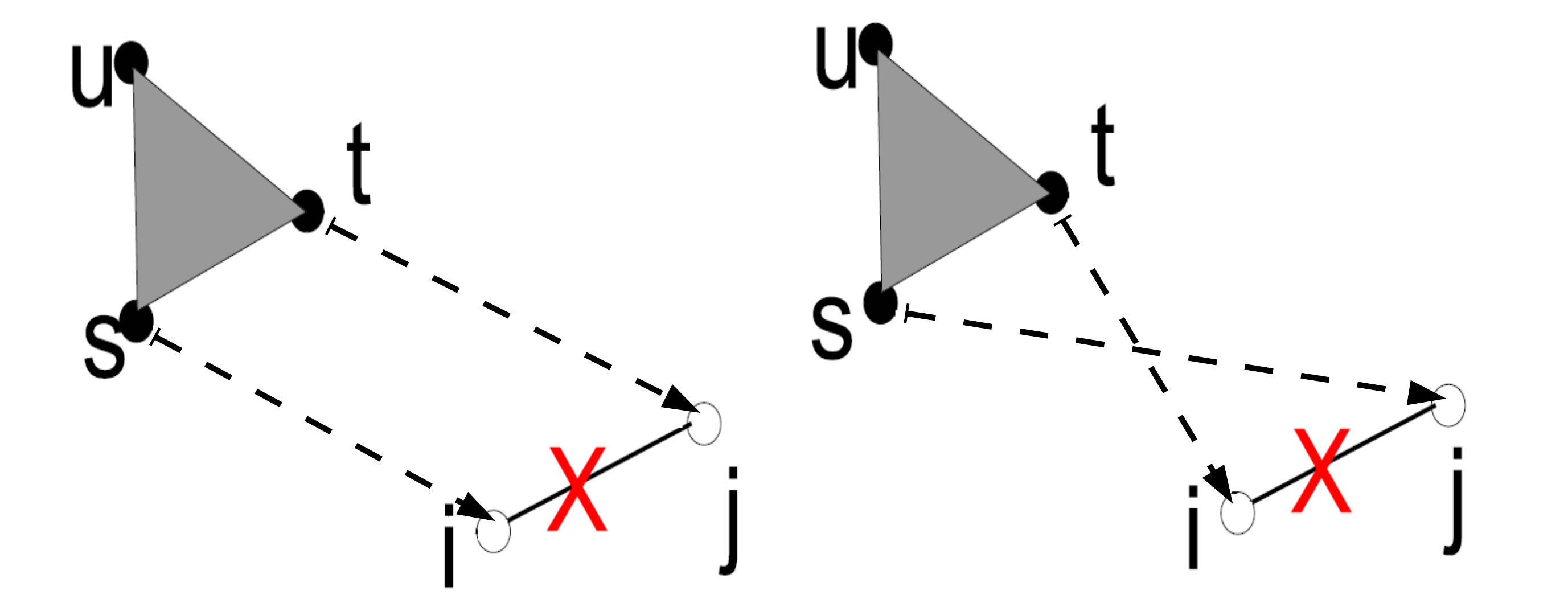}
\vspace{-0.25in}\caption{Multiple supernode merges resulting in subsumption}\label{fApp4}\vspace{-0.3in}
\end{figure}
%%%%%%%%%%%%%%%%%%
\begin{figure}[ht]
\centering
\includegraphics[width=1.0\columnwidth]{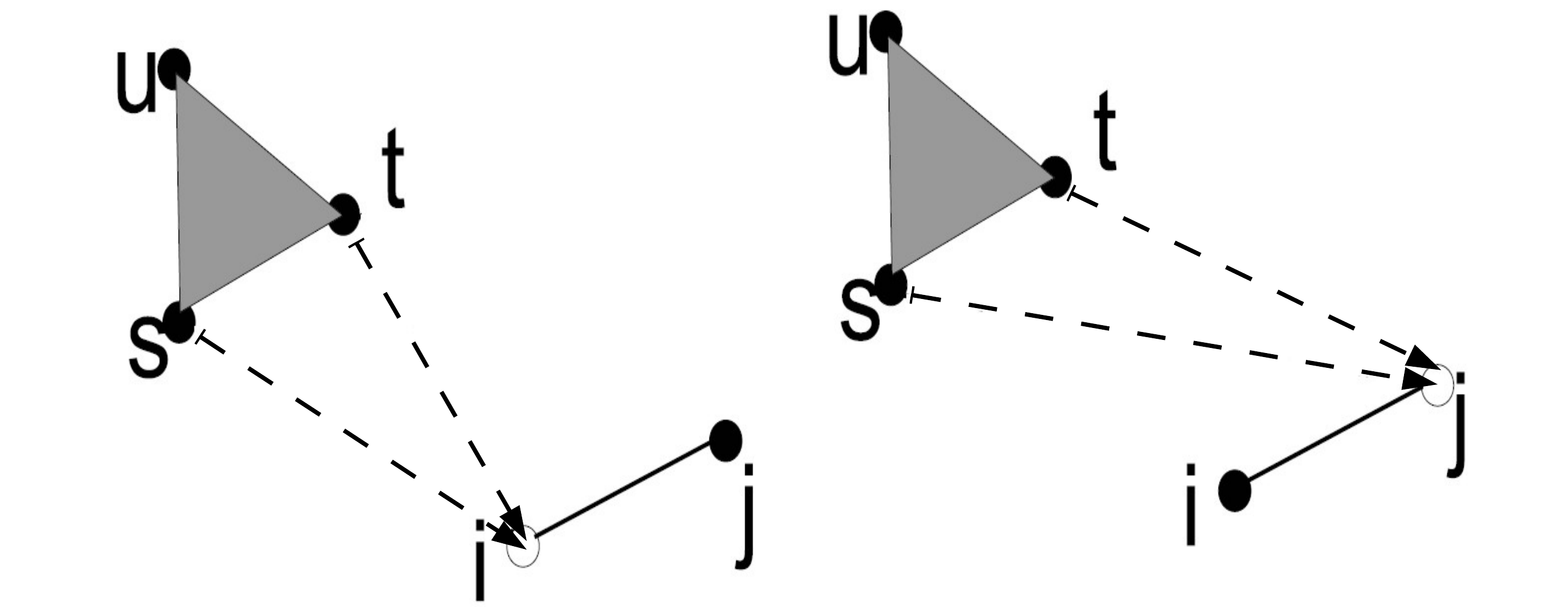}\vspace{-0.1in}
\caption{Multiple supernode virtual connections/merging selections not resulting in subsumption}\label{fApp5}\vspace{-0.1in}
\end{figure}
%%%%%%%%%%%%%%%%%%

\begin{lemma}\label{lem:subs-v-size}
The likelihood of subsumption of a facet reduces as facet size increases.
\end{lemma}
\begin{proof}
Note that for any facet with dimension $L$, the equivalent facet degree of the corresponding supernode $F_S$ is given by:
\begin{equation}
d_{S}=\sum_{i=1}^L d_i,
\end{equation}
which naturally implies that the equivalent facet degree of a supernode increases with the dimension of the originating facet. However, the number of multiple merges that is required for the subsumption for a facet of dimension $L$ is greater than or equal to $L$. Accordingly, this probability is $O(\frac{d_S^L}{f^L})$. Even though the supernode facet degree increases with $L$, which suggests a higher likelihood of multiple merges, the multiple merges which can subsume a facet decreases with facet size.

Note that on average, the equivalent facet degree $d_S$ increases linearly with facet dimension $L$. Let us assume that $d_S\approx \kappa L$, where $\kappa$ is a constant. Then, the likelihood of multiple merges to a supernode corresponding to a facet of dimension $L$ can be bounded as $O(\frac{\kappa^LL^L}{f^L})$.

On the other hand, recall from Lemma \ref{lem:subs-lessmult} that for a facet of dimension $L$, only a $\frac{L!}{L^L}$ fraction of the supernode multiple merges correspond to an actual subsumption event. Accordingly, the overall likelihood of subsumption can be approximated as:
\begin{equation}
\frac{\kappa^L L! L^L}{L^L f^L},
\end{equation}
which is maximized for $L=2$ (edge). This approximation is based on the assumption that all multiple edge merge events are of equal probability. In practice, while this is likely not the case, still $O(\frac{\kappa^LL^L}{f^L})$ decreases with increasing $L$, since $\frac{\kappa L}{f}<1$.% and events with a lower number of connections are more likely, suggesting that the probability of subsumption is largest for small facet sizes.
% and $L=3$ (filled triangle), and
%monotonically reduces afterwards.
\end{proof}

\begin{lemma}\label{lem:subs-gt1}
The probability of more than one facet being subsumed simultaneously is bounded from above by $O(\frac{(d_i+d_j+d_k)^3}{f^3})$, for $(i,j,k)=\arg\max_{i,j,k\in V}(d_i+d_j+d_k)$.
\end{lemma}
\begin{proof}
The minimum number of merges to a
 supernode which results in more than one facet subsumptions is
 equal to 3. This occurs when three edges connect with each other to
 form an empty triangle. Note that two edges sharing a common node
 also necessitates 3 merges by the incoming facet, but its
 supernode would have a lower facet degree. Since the likelihood of
 getting 3 merges is $O(\frac{d_S^3}{f^3})$, this quantity becomes vanishingly small as the SC grows. We also note that it is less likely that an incoming facet will get more than $3$ merges to a supernode.
\end{proof}

\begin{lemma}\label{lem:App6} The probability of $L$ facets being subsumed simultaneously
is upper bounded by the probability of a dimension-$L$ facet being subsumed. \end{lemma}
\begin{proof} 
Consider $L$ facets of dimension $L-1$ connecting such that they would form a facet of
dimension $L$ except that there is a \emph{hole}; e.g. three edges
connected to form an empty triangle, or four filled triangles
resulting in an empty tetrahedron. The probability of subsumption of
such combined structures of $L$ facets can be analyzed by using the
supernode technique, and can be shown to be negligible for large SCs.
\end{proof}

\begin{remark}\label{rem:App7} The supernode method can be also used for analyzing the probability of subsumption of incoming facets by the existing simplicial complex.
\end{remark}

\vspace{-2ex}

\section{Acknowledgments}
Research was sponsored by the Army Research Laboratory and was accomplished under Cooperative Agreement Number W911NF-09-2-0053. The views and conclusions contained in this document are those of the authors and should not be interpreted as representing the official policies, either expressed or implied, of the Army Research Laboratory or the U.S. Government. The U.S. Government is authorized to reproduce and distribute reprints for Government purposes notwithstanding any copyright notation here on. 

We would also like to thank Robert Drost (US Army Research Laboratory) for discussions on speeding up certain computational steps in the \GeneSCs algorithm, and Terrence J. Moore (US Army Research Laboratory) for feedback regarding the hypothetical example demonstrating an extreme case of subsumptions.

This document does not contain technology or technical data controlled under either the U.S. International Traffic in Arms Regulations or the U.S. Export Administration Regulations.

%%%%%%%%%%%%%%%%%%%%%%%%%%%%%

%%%%%%%%%%%%%%%%%%%%%%%%%%%%%
\bibliography{references}

\begin{thebibliography}{10}

\bibitem{collevo}
A.~L. Barabási, H.~Jeong, Z.~Néda, Ravasz E., A.~Schubert, and T.~Vicsek.
\newblock {Evolution of the social network of scientific collaborations. }.
\newblock {\em Physica A: Statistical mechanics and its applications}, 311(3),
  2002.

\bibitem{Chung2003}
F.~Chung, L.~Liu, T.~G. Dewey, and D.~J. Galas.
\newblock {Duplication Models for Biological Networks}.
\newblock {\em Journal of Computational Biology}, 10(5), 2003.

\bibitem{DMS2000}
S.~N. {Dorogovtsev}, J.~F.~F. {Mendes}, and A.~N. {Samukhin}.
\newblock {Structure of Growing Networks with Preferential Linking}.
\newblock {\em Physical Review Letters}, 85:4633--4636, November 2000.

\bibitem{distcol}
E.~Elmacioglu and D.~Lee.
\newblock {Modeling idiosyncratic properties of collaboration networks
  revisited.}
\newblock {\em Scientometrics}, 80(1), 2009.

\bibitem{Hatcher2002}
A.~Hatcher.
\newblock {\em {Algebraic Topology}}.
\newblock {Cambridge University Press}, Cambridge, England, 2002.

\bibitem{HAMND2011}
L.~H\'ebert-Dufresne, A.~Allard, V.~Marceau, P.-A. No\"el, and L.~J. Dub\'e.
\newblock Structural preferential attachment: Network organization beyond the
  link.
\newblock {\em Phys. Rev. Lett.}, 107:158702, Oct 2011.

\bibitem{HAMND2012}
L.~H\'ebert-Dufresne, A.~Allard, V.~Marceau, P.-A. No\"el, and L.~J. Dub\'e.
\newblock Structural preferential attachment: Stochastic process for the growth
  of scale-free, modular, and self-similar systems.
\newblock {\em Phys. Rev. E}, 85:026108, Feb 2012.

\bibitem{HoangRMS13}
M.~X. Hoang, R.~Ramanathan, T.~J. Moore, and A.~Swami.
\newblock Structural and collaborative properties of team science networks.
\newblock In {\em Advances in Social Networks Analysis and Mining 2013,
  {ASONAM} '13, Niagara, ON, Canada - August 25 - 29, 2013}, pages 1102--1109,
  2013.

\bibitem{HoangRS14}
M.~X. Hoang, R.~Ramanathan, and A.~K. Singh.
\newblock Structure and evolution of missed collaborations in large networks.
\newblock In {\em 2014 Proceedings {IEEE} {INFOCOM} Workshops, Toronto, ON,
  Canada, April 27 - May 2, 2014}, pages 849--854, 2014.

\bibitem{Krapivsky2001}
P.~L. Krapivsky and S.~Redner.
\newblock Organization of growing random networks.
\newblock {\em Phys. Rev. E}, 63:066123, May 2001.

\bibitem{Leskovec2007}
J.~Leskovec, J.~Kleinberg, and C.~Faloutsos.
\newblock {Graph Evolution: Densification and Shrinking Diameters}.
\newblock {\em ACM Trans. Knowl. Discov. Data}, 1(1), March 2007.

\bibitem{Liu2012}
D.~Liu, N.~Blenn, and P.~Van Mieghem.
\newblock Characterizing the structure of affliation networks.
\newblock {\em Procedia Computer Science}, 9(0):567 -- 576, 2012.
\newblock Proceedings of the International Conference on Computational Science,
  \{ICCS\} 2012.

\bibitem{clqc}
G.~R. Meleu and P.~M. Yonta.
\newblock {Growth model for collaboration networks. }.
\newblock {\em 2016. $<$hal-01304882$>$}.

\bibitem{colls}
M.~EJ Newman.
\newblock {The structure of scientific collaboration networks. }.
\newblock {\em Proceedings of the National Academy of Sciences}, 404(409),
  2001.

\bibitem{PDFV2005}
G.~Palla, I.~Derenyi, I.~Farkas, and T.~Vicsek.
\newblock {Uncovering the overlapping community structure of complex networks
  in nature and society}.
\newblock {\em Nature}, 435:814--818, June 2005.

\bibitem{unipahyb}
Flake~G.W. Lawrence S. Glover~E.J. Pennock, D.M. and C.L. Giles.
\newblock {Winners don't take all: Characterizing the competition for links on
  the web. }.
\newblock {\em Proceedings of the national academy of sciences}, 99(8), 2002.

\bibitem{slov}
M.~Perc.
\newblock {Growth and structure of Slovenia’s scientific collaboration
  network}.
\newblock {\em Journal of Informetrics}, 4(4), 2010.

\bibitem{Ramanathan2011}
R.~Ramanathan, A.~Bar-Noy, P.~Basu, M.~Johnson, W.~Ren, A.~Swami, and Q.~Zhao.
\newblock {Beyond Graphs: Capturing Groups in Networks}.
\newblock In {\em Proceedings of NetSciCom Workshop}, 2011.

\bibitem{Watts1998}
D.J. Watts and S.H. Strogatz.
\newblock Collective dynamics of 'small-world' networks.
\newblock {\em Nature}, (393):440--442, 1998.

\bibitem{Wu2015}
Z.~Wu, G.~Menichetti, C.~Rahmede, and G.~Bianconi.
\newblock {Emergent Complex Network Geometry}.
\newblock {\em Scientific Reports 5}, 2015.

\bibitem{npacol}
T.~Zhou, B.H. Wang, Y.D. Jin, D.R. He, P.P. Zhang, Y.~He, B.B. Su, K.~Chen,
  Z.Z. Zhang, and J.G Liu.
\newblock {Modelling collaboration networks based on nonlinear preferential
  attachment}.
\newblock {\em International Journal of Modern Physics C}, 18(2), 2007.

\end{thebibliography}
\bibliographystyle{plain}

%%%%%%%%%%%%%%%%%%%%%%%%%%%%%
\end{document}